\documentclass[a4paper,12pt]{amsart}

\usepackage{amsthm}
\usepackage{amsmath}
\usepackage{geometry}
\usepackage{amsfonts}
\usepackage{graphicx}
\usepackage{array}
\usepackage{amssymb}
\usepackage{mathtools}
\renewcommand{\d}{\mathrm{d}}

\newtheorem{theorem}{Theorem}

\newtheorem{proposition}[theorem]{Proposition}

\newtheorem{lemma}[theorem]{Lemma}

\numberwithin{equation}{section}
\numberwithin{theorem}{section}

\title[Hamiltonian Palatini's gravity with boundaries]{On a covariant Hamiltonian description of Palatini's gravity on manifolds with boundary}

\author{A. Ibort}
\address{ICMAT and Depto. de Matem\'aticas, Univ. Carlos III de Madrid, Avda. de la
Universidad 30, 28911 Legan\'es, Madrid, Spain.}
\email{albertoi@math.uc3m.es}
\author{A. Spivak}
\address{Dept. of Mathematics, Univ. of California at Berkeley, 903 Evans Hall, 94720 Berkeley CA, USA}
\email{ameliaspivak@gmail.com}
\begin{document}
\begin{abstract}
 A covariant Hamiltonian description of Palatini's gravity on manifolds with boundary is presented.  Palatini's gravity appears as a gauge theory satisfying a constraint in a certain topological limit. This approach allows the consideration of non-trivial topological situations.\\  
 The multisymplectic framework for first-order covariant Hamiltonian field theories on manifolds with boundary, developed in [Ib15], enables analysis of the system at the boundary. The reduced phase space of the system is determined to be a symplectic manifold with a distinguished isotropic submanifold corresponding to the boundary data of the solutions of the Euler-Lagrange equations.  
\end{abstract}

\maketitle
\tableofcontents

%%%%%%%%%%%%%%%%%%%%%%%%%%%%%%%%%%%%%%%%%%%%%%%%%%%%%%%%%%%%%%%%%%%%%%%%%%%%%%%%%%%%%%%%%%%%%%%%%%%%%%%%%%%%%%%%%%%%%%%%%%%%

\section{Introduction}\label{sec:introduction}

 Our understanding of the Hamiltonian structure of Gravity has taken half a century. The initial difficulties faced by Dirac and Bergmann [Be58],[Be81], were slowly resolved through the work of Arnowit, Deser and Misner [Ar62], all the way to Ashtekar's formulation [As87]. At least part of the motivation has been to place the theory of gravity on grounds that will make it suitable for a canonical quantization scheme.

In [Ro06], C. Rovelli illustrated a simple Hamiltonian formulation of General Relativity which is manifestly 4D generally covariant and that drops the reference to the underlying space-time in Palatini's formulation of gravity.
Rovelli's proposal is highly geometrical and constructs its space as the $4+16+24$ dimensional space $\widetilde{\mathcal{C}}$ with local coordinates $(x^\mu, e_\mu^I, A_\mu^{IJ})$.  In a further effort at extracting the geometrical essence of such space, the variables $x^\mu$ are dropped (accounting by the invariance of the theory under global diffeomorphisms) and we are led to a $40$ dimensional space $\mathcal{C}$ [Ro01].  The disappearance of the spacetime manifold $M$ and its coordinates $x^\mu$, which survive only as arbitrary parameters on the `gauge orbits' of the canonical geometrical structure defined on it, generalizes the disappearance of the time coordinate in the ADM formalism and is analogous to the disappearance of the Lagrangian evolution parameter in the Hamiltonian theory of a free particle [Ro01]. It simply means that the general relativistic space- time coordinates are not directly related to observations.

Our program in this paper is similar but our inspiration is the geometrical foundations of covariant first order Hamiltonian field theories on manifolds with boundary discussed recently in [Ib15].   There the role of a covariant phase space for a first order Hamiltonian theory modelled on the affine dual space of the first jet bundle of the bundle defining the fields of the theory is assessed and the crucial role played by boundaries as determining symplectic spaces of fields defining the classical counterpart of the quantum states of the theory is stressed in accordance with the point of view expressed in [Sc51].   

Actually a generally covariant notion of instantaneous state, or evolution of states and observables in time, make little physical sense.   They are always referred to an initial data space-like surface that in the picture presented here, corresponds to the boundary of the space-times of events.    Such notion does not really conflicts with diffeomorphism invariance because a diffeomorphism of a smooth manifold with smooth boundary restricts to a diffeomorphism of the boundary.   Thus, providing that the notion of boundary of a spacetime is incorporated in the basic description of the theory, we may still consider diffeomorphism invariance as a fundamental notion without contradicting it.    

The covariant phase space of the theory carries a natural multisymplectic structure which is the exterior differential of a canonical $m$-form $\Theta$ defined on it.   This geometrical structure has been considered in various guises in the various variational formulations of field theories, however its first use in the present setting is to help to identify the nature of the different fields of the theory.   Thus it will be discussed how the vierbein fields  $e_\mu^I$ correspond to an algebraic constraint imposed in the momenta fields of the theory.    The corresponding action will be seen to be invariant under the group of all automorphisms of the geometrical structure and it will induce the corresponding reduction on the space of gauge fields at the boundary.   This reduction process is interpreted as the appropriate setting for the `elimination' of the space-time $M$, i.e., the space of physical classical solutions  of the theory in the bulk is the moduli space of the space of solutions of the Euler-Lagrange equations with respect to the 
group of automorphisms whereas, the phase space of physical degrees of freedom of the theory, associated to its boundary, is the reduced symplectic manifold of fields at the boundary.    

We can give $\mathcal{C}$ a direct physical interpretation in terms of reference systems transformations. In the quantum domain, it leads directly to the spin-network to spin-network amplitudes computed in loop quantum gravity.

%%%%%%%%%%%%%%%%%%%%%%%%%
%%%%%%%%%%%%%%%%%%%%%%%%%
%%%%%%%%%%%%%%%%%%%%%%%%%

\section{The geometry of the covariant phase space for Yang-Mills theories}

As discussed in the introduction our approach to Palatini's gravity will be to consider it as a constrained first order covariant Hamiltonian field theory on a manifold with boundary obtained as a topological phase of a gauge theory.     We will review first the geometrical setting for covariant first order Hamiltonian Yang-Mills theories and the topological phase that will interest us.   

\subsection{A brief account of the multisymplectic formalism for first order covariant Hamiltonian Yang-Mills theories on manifolds with boundary}\label{sec:general}

%%%%%%%%%%%%%%%%%%%%%%%%%%
%%%%%%%%%%%%%%%%%%%%%%%%%%

We will review first the basic notions and notations for first order covariant Hamiltonian field theories (see more details in [Ib15]).

\subsubsection{The covariant phase space of Yang-Mills theories}

The fundamental geometrical structure of a given first order Hamiltonian theory will be provided by a fiber bundle $\pi \colon E \to M$ with $M$ an $m = (1+d)$-dimensional orientable smooth manifold with smooth boundary $\partial M \neq \emptyset$ and local coordinates adapted to the fibration $(x^\mu, u^a)$, $a= 1, \ldots, r$, where $r$ is the dimension of the standard fiber.  Because $M$ is orientable we will assume that a given volume form $\mathrm{vol}_M$ is selected.  Notice that it is always possible to chose local coordinates $x^\mu$ such that $\mathrm{vol}_M = dx^0 \wedge dx^1 \wedge \cdots \wedge dx^d$.

Yang-Mills fields are principal connections $A$ on some principal fiber bundle $P \to M$ with structural group $G$.   For clarity in the exposition we are going to make the assumption that $P$ is trivial (which is always true locally), i.e., $P \cong M \times G  \to M$ where (again, for simplicity) $G$ is a Lie group with Lie algebra $\mathfrak{g}$.  Under these assumptions, principal connections on $P$ can be identified with $\mathfrak{g}$-valued 1-forms on $M$, i.e., with sections of the bundle $E=T^*M \otimes \mathfrak{g} \longrightarrow M$.  Local bundle coordinates in the bundle $E \to M$ will be written as $(x^\mu, A_\mu^a)$, $\mu = 1, \ldots, m$, $a= 1, \ldots, \dim\mathfrak{g}$, where 
$A = A_\mu^a \, \xi_a \in \mathfrak{g}$ with ${\xi}_a$ a basis of the Lie algebra $\mathfrak{g}$.  Thus, a section of the bundle can be written as 
\begin{equation}\label{connectionA}
A(x) = A^a_{\mu}(x) {dx}^{\mu}{\otimes}{\xi}_a \, .
\end{equation}

We will denote by $\pi_1^0\colon J^1E \to E$ the affine 1-jet bundle of the bundle $E\stackrel{\pi}{\rightarrow}M$.
The elements of $J^1E$ are equivalence classes of germs of sections $\phi$ of $\pi$, i.e., two sections $\phi$, $\phi'$ at $x\in M$ are equivalent, i.e., represent the same germ,  if $\phi(x) = \phi'(x)$ and $d\phi (x) = d\phi'(x)$.   The bundle $J^1E$ is an affine bundle over $E$ modeled on the vector bundle $VE \otimes_E \pi^* (T^*M)$.   If $(x^\mu; u^a)$, $\mu = 0, \ldots, d$ is a bundle chart for the bundle $\pi\colon E\to M$, then we will denote by $(x^\mu,u^a; u^a_\mu)$ a local chart for the jet bundle $J^1E$.
So in the case of Yang-Mills, local coordinates  on $J^1E$ will be denoted by $(x,A^a,A^a_{\mu})$.

The affine dual of $J^1E$ is the vector bundle over $E$ whose fiber at $\xi = (x,u)$ is the linear space of affine maps $\mathrm{Aff}(J^1E_\xi, \mathbb{R})$.   The vector bundle $\mathrm{Aff}(J^1E, \mathbb{R})$, possesses a natural subbundle defined by constant functions along the fibers of $J^1E \to E$, that we will denote again, with an abuse of notation, as $\mathbb{R}$, then the quotient bundle $\mathrm{Aff}(J^1E, \mathbb{R})/\mathbb{R}$ will be called the covariant phase space bundle of the theory, or the phase space for short.   Notice that such bundle, denoted in what follows by $P(E)$ is the vector bundle with fibre at $\xi = (x,u) \in E$ given by $(V_uE\otimes T_x^*M)^* \cong T_xM \otimes(V_uE)^*\cong \mathrm{Lin}(V_uE,T_xM)$ and projection $\tau_1^0 \colon P(E) \to E$.

Local coordinates on $P(E)$ can be introduced as follows:
Affine maps on the fibers of $J^1E$ have the form $u_{\mu}^a \mapsto \rho_0 + \rho_a^{\mu}u_{\mu}^a$ where $u_{\mu}^a$ are natural coordinates on the fiber over the point $\xi$ in $E$ with coordinates $(x^{\mu},u^a)$. Thus an affine map on each fiber over $E$ has coordinates $\rho_0, \rho^{\mu}_a$, with $\rho^\mu_a$ denoting linear coordinates on $TM \otimes VE^*$ associated to bundle coordinates $(x^\mu, u^a)$.   Functions constant along the fibers are described by the numbers $p_0$, hence elements in the fiber of $P(E)$ have coordinates $\rho_a^{\mu}$.  Thus a bundle chart for the bundle $\tau_1^0\colon P(E) \to E$ is given by $(x^\mu, u^a;  \rho^\mu_a)$.

The choice of a distinguished volume form $\mathrm{vol}_M$ in $M$ allows us to identify the fibers of $P(E)$ with a subspace of $m$-forms on $E$ as follows ([Ca91]):
The map $u_{\mu}^a \rightarrow \rho_a^{\mu}u_{\mu}^a$ corresponds to the $m$-form $\, \, \rho_a^\mu d u^a\wedge \mathrm{vol}_{\mu}$
where vol$_{\mu}$ stands for $i_{{\partial}/{\partial x^{\mu}}}$vol$_M.$
Let ${\bigwedge}^m (E)$ denote the bundle of $m$-forms on $E$. Let 
${\bigwedge}_k^m(E)$ be the subbundle of ${\bigwedge}^m (E)$ consisting of those $m$-forms which vanish when $k$ of their arguments are vertical.  So in our local coordinates, elements of ${\bigwedge}_1^m(E)$, i.e., $m$-form on $E$ that vanish when one of their arguments is vertical, commonly called semi-basic 1-forms, have the form $\rho_a^{\mu} d u^a\wedge \mathrm{vol}_\mu + \rho_0 \mathrm{vol}_M$, and elements of ${\bigwedge}_0^m(E)$, i.e., basic $m$-forms, have the form $p_0\mathrm{vol}_M$. These bundles form a short exact sequence:
$$
0\rightarrow\textstyle{\bigwedge}^m_0E\hookrightarrow
	\textstyle{\bigwedge}^m_1E\rightarrow P(E)\rightarrow0 \, .
$$
Hence ${\bigwedge}_1^m E$ is a real line bundle over $P(E)$ and, for each point $\zeta = (x,u,p)\in P(E)$, the fiber is the quotient $\bigwedge_1^m (E) _\zeta/ \bigwedge_0^m (E)_\zeta$.

In the case of Yang-Mills, elements of $P(E)$ have the form $P = P_a^{\mu \nu}dA_{\mu}^a \wedge d^{m-1}x_{\nu}$.
		
  The bundle $\textstyle{\bigwedge}_1^m(E)$ carries a canonical $m$--form  which may be defined by a generalization of the definition of the canonical 1-form on the cotangent bundle of a manifold.  Let $\sigma \colon \textstyle{\bigwedge}_1^m(E) \to E$  be the canonical projection, then the canonical $m$-form $\Theta$ is defined by 
$$
\Theta_\varpi(U_1,U_2,\ldots,U_m) = \varpi(\sigma_*U_1, \ldots, \sigma_*U_m)
$$
where $\varpi\in\bigwedge^m_1(E)$ and $U_i\in T_\varpi(\bigwedge^m_1(E))$.
As described above, given bundle coordinates $(x^\mu,u^a)$ for $E$
we have coordinates $(x^\mu,u^a,\rho, \rho^\mu_a)$
on $\bigwedge^m_1(E)$ adapted to them  
and the point $\varpi\in\bigwedge^m_1(E)$ with coordinates
$(x^\mu,u^a;\rho, \rho^\mu_a )$ is the $m$-covector
$\varpi =  \rho^\mu_a\, d u^a\wedge \mathrm{vol}_\mu + \rho \, \mathrm{vol}_M$.
With respect to these same coordinates we have the local expression
$$
\Theta =  \rho^\mu_a \, d u^a \wedge \mathrm{vol}_\mu  + \rho\, \mathrm{vol}_M \, ,
$$
for $\Theta$, where $\rho$ and $\rho^\mu_a$ are now to be interpreted as coordinate
functions.
	
The $(m+1)$-form $\Omega = d \Theta $ defines a multisymplectic structure on the manifold $\bigwedge_1^m(E)$, i.e.$(\bigwedge^m_1(E),\Omega)$ is a multisymplectic manifold. There is some variation in the literature on the definition of multisymplectic manifold. For us, following [Ca91] and [Go98], a multisymplectic manifold is a pair $(X,\Omega)$ where X is a manifold of some dimension m and $\Omega$ is a form on $X$ of some dimension $d \geq 2$, and $\Omega$ is both closed and nondegenerate. By nondegenerate we mean that if $i_v{\Omega} = 0$ then $v=0$.
	
We will refer to 
$\bigwedge^m_1E$ by $M(E)$ to emphasize that its status as a multisymplectic manifold. We will denote the projection $M(E) \to E$
by $\nu$, while the projection $M(E) \to P(E)$ will be
denoted by $\mu$.  Thus $\nu = \tau^0_1\circ\mu$, with $\tau_1^0 \colon P(E) \to E$
the canonical projection.(See figure 1.)

A Hamiltonian $H$ on $P(E)$ is a section of
$\mu$. Thus in local coordinates 
$$
H(\rho_a^{\mu}\, d u^a\wedge\mathrm{vol}_{\mu}) = \rho_a^{\mu} d u^a\wedge \mathrm{vol}_{\mu}-\mathbf{H}(x^{\mu},u^a, \rho_a^{\mu}) \mathrm{vol}_M \, ,
$$ 
where $\mathbf{H}$ is here a real-valued function.
	
We can use the Hamiltonian section $H$ to define an $m$-form on $P(E)$
by pulling back the canonical $m$-form $\Theta$ from $M(E)$.  We call
the form so obtained the Hamiltonian $m$-form associated with $H$ and denote
it by $\Theta_H$. Thus if we write the section defined in local coordinates $(x^\mu, u^a;\rho, \rho_a^\nu )$ as 
\begin{equation}\label{rhoH}
\rho = - \mathbf{H}(x^{\mu}, u^a, \rho_a^\mu ) \, ,
\end{equation}
then
\begin{equation}\label{ThetaH}
	\Theta_H = \rho_a^\mu\, d u^a \wedge \mathrm{vol}_\mu - \mathbf{H}(x^\mu,u^a, \rho_a^\mu) \, \mathrm{vol}_M \, .
\end{equation}
\,
 In (2.1) the minus sign in front of the Hamiltonian is chosen to be in keeping with the traditional conventions in mechanics for the integrand of the action over the manifold: $pdq - Hdt$. When the form $\Theta_H$ is pulled back to the manifold $M$, as described in section 2.2.1, the integrand of the action over $M$ will have a form reminiscent of that of mechanics, with a minus sign in front of the Hamiltonian. See equation (2.5).
 
\subsubsection{The action and the variational principle}
	From here on, in addition to being an oriented smooth manifold with either a Riemannian or a Lorentzian metric, $M$ has a boundary $\partial M$. The orientation chosen on $\partial M$ is consistent with the orientation on $M$. Everything in the last section applies. The presence of boundaries will enable us to enlarge the use to which the multisymplectic formalism can be applied, starting with eqn. $(2.5)$.	

The fields $\chi$ of the theory in the Hamiltonian formalism constitute a class of sections of the bundle $\tau_1 : P(E) \to M$.    $P(E)$ is a bundle over $E$ with projection $\tau_1^0$ and it is a bundle over $M$ with projection $\tau_1 = \pi\circ \tau_1^0$.  The sections that will be used to describe the classical fields in the Hamiltonian formalism are those sections $\chi\colon M \to P(E)$ ,i.e.  $\tau_1\circ \chi  = \mathrm{id}_M$,  such that $\chi = P \circ \Phi$ where $\phi \colon M \rightarrow E$ is a section of $\pi: E\rightarrow M$, i.e. $\pi \circ \Phi = \mathrm{id}_M$, and  $P \colon E \to P(E)$, is a section of $\tau_1^0 \colon P(E) \rightarrow E$ i.e. $\tau_1^0 \circ P = \mathrm{id}_P$. (See Figure).  The sections $\Phi$ will be called the configurations and the sections $P$ the momenta of the theory.   In other words $u^a = \Phi^a(x)$ and $\rho_a^\mu = P_a^\mu (\Phi(x))$ will provide local expression for the section $\chi = P \circ \Phi$.
 We will denote such a section $\chi$ by $(\Phi, P)$ to indicate the iterated bundle structure  of $P(E)$ and we will refer to $\chi$ as a double section\footnote{It can also be said that $\chi$ is a section of $P(E)$ along $\Phi$.}.
 
 \begin{figure}[ht]
\centering
\includegraphics[width=8cm]{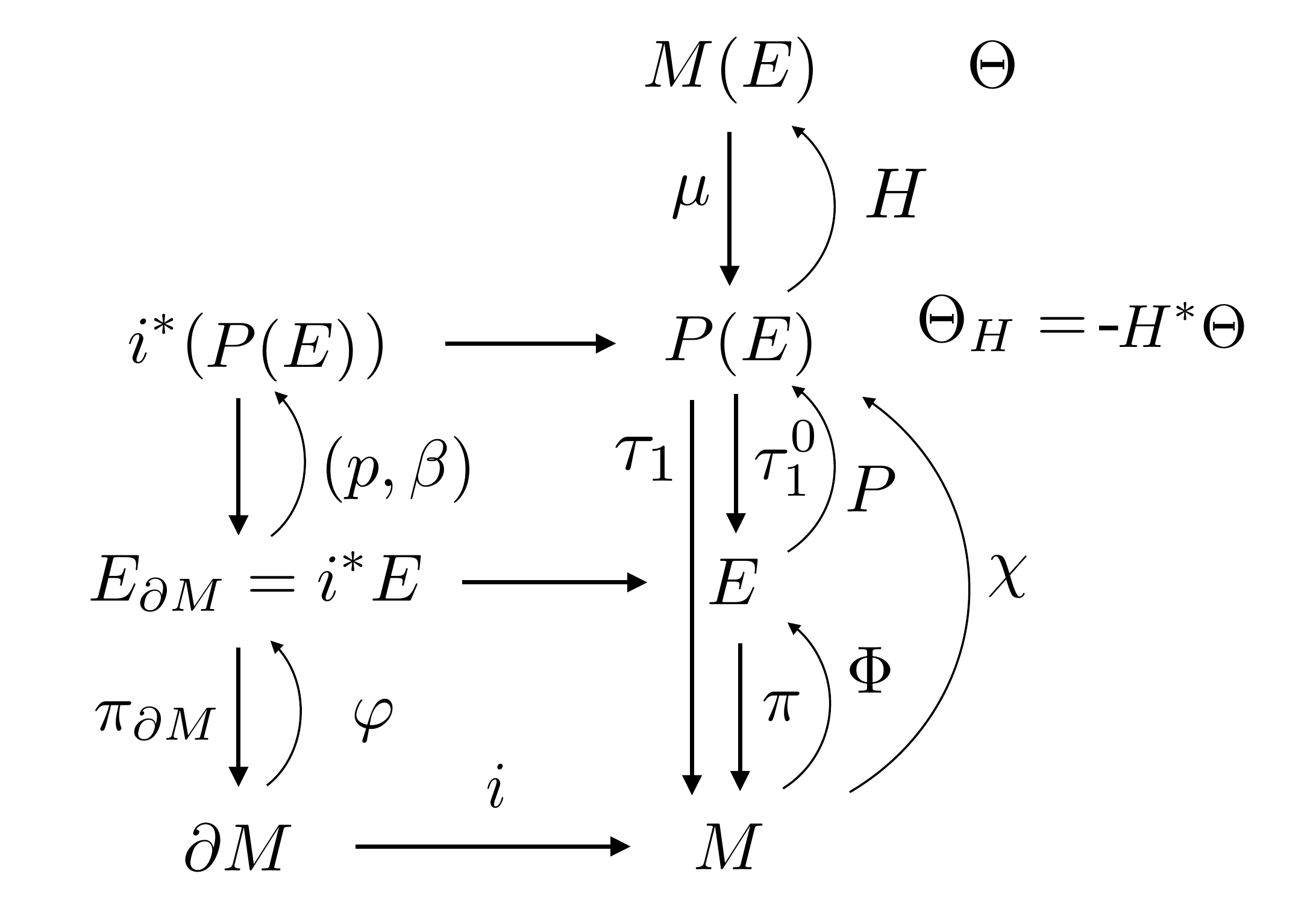}
\caption{Bundles, sections and fields: configurations and momenta}
\end{figure}

We will denote by $\mathcal{F}_M$ the space of sections $\Phi$ of the bundle $\pi \colon E \rightarrow M $, that is $\Phi \in \mathcal{F}_M$, and we will denote by $\mathcal{F}_{P(E)}$ the space of double sections $\chi = (\Phi, P)$.  Thus $\mathcal{F}_{P(E)}$  represents the space of fields of the theory in the first order covariant Hamiltonian formalism.

 Thus the fields of the theory in the multisymplectic picture for Yang-Mills theories are provided by sections $(A,P)$ of the double bundle $P(E) \to E \to M$.

The equations of motion of the theory will be defined by means of a variational principle, i.e., they will be characterized as the critical points of an action functional $S$ on $\mathcal{F}_{P(E)}$. Such action will be given simply by
\begin{equation}\label{action}
 S(\chi ) = \int_M \chi^*\Theta_H  ,
 \end{equation}
 In the case of Yang-Mills theories, the action in a first-order covariant Hamiltonian formulation of the theory is given by,
 \begin{equation}\label{ymap}
S_{\mathrm{YM}}(A,P) = \int_M P_a^{\mu\nu} dA_\mu^a \wedge d x^{m-1}_\nu - H_{\lambda}(A,P) \mathrm{vol}_M \, .
\end{equation}
with Hamiltonian function,
\begin{equation}\label{hamiltonian}
H_{\lambda}(A,P) = \frac{1}{2} \epsilon_{bc}^aP^{\mu\nu}_a A_\mu^bA_\nu^c + \frac{\lambda}{4}P^{\mu\nu}_a P_{\mu\nu}^a \, 
\end{equation}
for some ${\lambda} \geq 0$,
where the indexes $\mu\nu$  ($a$) in $P_a^{\mu\nu}$ have been lowered (raised) with the aid of the Lorentzian metric $\eta$ (the Killing-Cartan form on $\mathfrak{g}$, respect.).

 Of course, as is usual in the derivations of equations of motion via variational principles, we assume that the integral in Eq. $(2.4)$ is well defined.  It is also assumed that the `differential' symbol in equation $(2.7)$ below, defined in terms of directional derivatives, is well defined and that the same is true for any other similar integrals that will appear in this work.

A simple computation leads to,
\begin{equation}\label{dSfirst}
\mathrm{d} S (\chi) (U) = \int_M \chi^* \left(i_{\widetilde U} d\Theta_H \right) + \int_{\partial M} (\chi\circ i)^* \left(i_{\widetilde U} \Theta_H \right) \, , 
\end{equation}
where $U$ is a vector field on $P(E)$ along the section $\chi$, $\widetilde{U}$ is any extension of $U$ to a tubular neighborhood of the image of $\chi$, and $i\colon \partial M \to M$ is the canonical embedding.

%%%%%%%%%%%%%%%%%%
%%%%%%%%%%%%%%%%%%

\subsubsection{The cotangent bundle of fields at the boundary}\label{sec:cotangent_boundary}

The boundary term $\int_{\partial M} (\chi\circ i)^* \left(i_{\tilde U} \Theta_H\right)$ in eq. (2.7) suggests that there is a family of fields at the boundary that play a special role.  Actually, we notice that the field $\tilde{U}$ being vertical with respect to the projection $\tau_1\colon P(E) \to M$ has the form $\tilde{U} = A^a \,  \partial/\partial u^a + B_\mu^a  \, \partial/\partial \rho_\mu^a$. Hence we obtain for the boundary term,
\begin{equation}\label{boundaryfirst}
\int_{\partial M} (\chi\circ i)^* \left(i_{\widetilde U} \Theta_H \right) = \int_{\partial M} (\chi\circ i)^*  \rho_a^\mu \, A^a  \, \mathrm{vol}_\mu  = \int_{\partial M} i^*(P_a^\mu\, A^a \, \mathrm{vol}_\mu)
\end{equation}
for $\chi = (\Phi, P)$. 

We will assume that there exists a collar around the boundary $U_\epsilon \cong (-\epsilon, 0] \times \partial M$, and we choose local coordinates  $x^0 = t \in (-\epsilon, 0]$, and $x^k$, $k = 1, \ldots, d$, describing local coordinates for $\partial M$,  such that $\mathrm{vol}_{U_\epsilon} =  d t \wedge \mathrm{vol}_{\partial M}$. The r.h.s. of eq. (2.8) then becomes,
\begin{equation}\label{boundary_final}
\int_{\partial M} i^*(P_a^\mu\, A^a \, \mathrm{vol}_\mu)  = \int_{\partial M} p_a \, A^a \, \mathrm{vol}_{\partial M} \, ,
\end{equation}
where $p_a = P_a^0\circ i$ is the restriction to $\partial M$ of the zeroth component of the momenta field $P_a^\mu$.

 Consider the space of fields at the boundary obtained by restricting the zeroth component of sections $\chi$ to $\partial M$, that is the fields of the form (see Figure 1)
$$
\varphi^a = \Phi^a \circ i \, , \qquad p_a = P_a^0 \circ i \, .
$$
Notice that the fields $\varphi^a$ are nothing but sections of the bundle $i^*E$, the pull-back along $i$ of the bundle $E$, while the space of fields $p_a$ can be thought of as 1-semibasic $d$-forms on $i^*E \to \partial M$.   This statement is made precise in the following:

\begin{lemma}\label{decomposition}   Given a collar around $\partial M$, $U_\epsilon \cong (-\epsilon, 0] \times \partial M$ and a volume form $\mathrm{vol}_{\partial M}$ on $\partial M$ such that $\mathrm{vol}_{U_{\epsilon}} = d t \wedge \mathrm{vol}_{\partial M}$ with $t$ the normal coordinate in $U_\epsilon$, then the pull-back bundle $i^*(P(E))$ is a bundle over the pull-back bundle $i^*E$ and decomposes 
naturally as $i^*P(E) \cong \bigwedge_1^m(i^*E) \oplus \bigwedge_1^{m-1}(i^*E)$.
\end{lemma}

\begin{proof} See[Ib15]
\end{proof}

If we denote by $\mathcal{F}_{\partial M}$ the space of configurations of the theory, $\varphi^a$, i.e., $\mathcal{F}_{\partial M} = \Gamma(i^*E)$, then the space of momenta of the theory $p_a$ can be identified with the space of sections of the bundle $\bigwedge_1^m(i^*E) \to i^*E$, according to Lemma 2.1.   Therefore the space of fields $(\varphi^a, p_a)$ can be identified with the cotangent bundle $T^*\mathcal{F}_{\partial M}$ over $\mathcal{F}_{\partial M}$ in a natural way, i.e., each field $p_a$ can be considered as the covector at $\varphi^a$ that maps the tangent vector $\delta\varphi^a$ at $\varphi^a$ into the number $\langle p, \delta \varphi \rangle$ given by,
\begin{equation}\label{pairing_cotangent}
\langle p, \delta \varphi \rangle = \int_{\partial M} p_a(x)\delta\varphi^a (x) \mathrm{vol}_{\partial M} \, .
\end{equation}
	
Notice that the tangent vector $\delta \varphi$ at $\varphi$ is a vertical vector field on $E$ along $\varphi$, and the section $p$ is a 1-semibasic $m$-form on $E$ (Lemma 2.1). Hence the contraction of $p$ with $\delta\varphi$ is an $(m-1)$-form along $\varphi$, and its pull-back $\varphi^*\langle p, \delta\varphi \rangle$ along $\varphi$ is an $(m-1)$-form on $\partial M$ whose integral defines the pairing above, Eq. (2.10). 	
	
 Viewing the cotangent bundle $T^*\mathcal{F}_{\partial M}$ as double sections $(\varphi, p)$ of the bundle $\bigwedge_1^m(i^*E)$ described by Lemma 2.1, the canonical 1-form $\alpha$ on $T^*\mathcal{F}_{\partial M}$ can be expressed as,
\begin{equation}\label{alpha}
\alpha_{(\varphi, p)} (U) = \int_{\partial M} p_a (x) \delta\varphi^a (x) \, \mathrm{vol}_{\partial M}
\end{equation}
where $U$ a tangent vector to $T^*\mathcal{F}_{\partial M}$ at $(\varphi, p)$, that is, a vector field on the space of 1-semibasic forms on $i^*E$ along the section $(\varphi^a, p_a)$, and therefore of the form $U = \delta\varphi^a \, \partial /\partial u^a + \delta p_a \, \partial /\partial \rho_a$.
	
Finally, notice that the pull-back to the boundary map $i^*$, defines a natural map from the space of fields in the bulk, $\mathcal{F}_{P(E)}$, into the phase space of fields at the boundary $T^*\mathcal{F}_{\partial M}$.  Such map will be denoted by $\Pi$ in what follows, that is, 
$$
\Pi \colon \mathcal{F}_{P(E)}\to T^*\mathcal{F}_{\partial M} \, , \qquad \Pi(\Phi, P) = (\varphi, p) , \, \quad  \varphi = \Phi\circ i, \, p_a = P_a^0\circ i \, .
$$

With the notations above, by comparing the expression for the boundary term given by eq. \ref{boundary_final}, and the expression for the canonical 1-form $\alpha$, eq. (2.11), we obtain,
$$
\int_{\partial M} (\chi\circ i)^* \left(i_{\tilde U} \Theta_H\right) = (\Pi^*\alpha)_\chi (U) \, .
$$
In words, the boundary term in eq.(2.7) is just the pull-back of the canonical 1-form $\alpha$ at the boundary along the projection map $\Pi$.

In what follows it will be customary to use the variational derivative notation when dealing with spaces of fields. For instance, if $F(\varphi,p)$ is a differentiable function defined on $\mathcal{F}_{\partial M}$ we will denote by $\delta F / \delta \varphi^a$ and $\delta F / \delta p_a$ functions (if they exist) such that
\begin{equation}\label{dF}
 d F_{(\varphi,p)}(\delta \varphi^a, \delta p_a) = \int_{\partial M} \left( \frac{\delta F}{\delta \varphi^a} \delta \varphi^a + \frac{\delta F}{\delta p_a} \delta p_a \right) \mathrm{vol}_{\partial M} \, ,
\end{equation}
with $U = (\delta \varphi^a, \delta p_a)$ a tangent vector at $(\varphi,p)$. 
We also use an extended Einstein summation convention such that integral signs will be omitted when dealing with variational differentials. For instance,
\begin{equation}\delta F =   \frac{\delta F}{\delta \varphi^a} \delta \varphi^a + \frac{\delta F}{\delta p_a} \delta p_a \, ,
\end{equation} 
may replace $d F$ as in Eq. (2.12).   Also in this vein we will write,
$$
\alpha = p_a \, \delta \varphi^a \, ,
$$
and the canonical symplectic structure $\omega_{\partial M} = -d \alpha$ on $T^*\mathcal{F}_{\partial M}$ will be written as,
$$
\omega_{\partial M} = \delta \varphi^a \wedge \delta p_a \, ,
$$
by which we mean
$$
\omega_{\partial M} ((\delta_1\varphi^a, \delta_1p_a), (\delta_2\varphi^a, \delta_2p_a)) = \int_{\partial M} \left( 
\delta_1\varphi^a(x) \delta_2 p_a(x) - \delta_2\varphi^a (x) \delta_1p_a(x) \right) \mathrm{vol}_{\partial M} \, ,
$$
where $(\delta_1\varphi^a, \delta_1p_a), (\delta_2\varphi^a, \delta_2p_a)$ are two tangent vectors at $(\varphi^a, p_a)$.

%%%%%%%%%%%%%%%%%%
%%%%%%%%%%%%%%%%%%

\subsubsection{Euler-Lagrange's equations and Hamilton's equations}
We now examine the contribution from the first term in $ d S$, eq. (2.7).
Notice that such a term can be 
thought of as a 1-form on the space of fields on the bulk, $\mathcal{F}_{P(E)}$.  We will call it the Euler-Lagrange 1-form and denote it by $\mathrm{EL}$, thus with the notation of eqn $(2.7)$,
$$
\mathrm{EL}_\chi (U) = \int_M \chi^* \left(i_{\tilde U} \d \Theta _H \right)  \, .
$$
A double section $\chi = (\Phi, P)$ of $P(E) \to E \to M$ will be said to satisfy the Euler-Lagrange equations determined by the first-order Hamiltonian field theory defined by $H$, if $\mathrm{EL}_\chi = 0$, that is, if $\chi$ is a
zero of the Euler-Lagrange 1-form $\mathrm{EL}$ on $\mathcal{F}_{P(E)}$.   Notice that this is equivalent to 
\begin{equation}\label{formEL}
	\chi^*(i_{\tilde{U}} d \Theta _H)=0 \, ,
\end{equation}
for all vector fields $\tilde{U}$ on a tubular neighborhood of the image of $\chi$ in $P(E)$. 
The set of all such solutions of Euler-Lagrange equations will be denoted by $\mathcal{EL}_M$ or just $\mathcal{EL}$ for short.\
 
 If the metric $\eta$ on $M$ is just the Minkowski metric so that $\sqrt{|\eta|} = 1$ or if we change to normal coordinates on $M$ which we can always find, then the volume element takes the form $\mathrm{vol_M}= dx^0 \wedge\cdots \wedge dx^d$. For local coordinates $(x^\mu,u^a,\rho^\mu_a)$ on  $P(E)$, using eqs. \eqref{rhoH}, \eqref{ThetaH}, we then have,
\begin{eqnarray*}
	i_{\partial/\partial \rho^\mu_a}\d \Theta _H &=& -\frac{\partial
	 H}{\partial \rho^\mu_a} d^mx + \d u^a\wedge \d^{m-1}x_\mu  \\
	i_{\partial/\partial u^a} \d \Theta _H &=& -\frac{\partial
	 H}{\partial u^a} d^mx - d\rho^\mu_a\wedge \d^{m-1}x_\mu.
\end{eqnarray*}
Applying (2.13) to these last two equations we obtain the Hamilton equations for the field in the bulk\footnote{The equations obtained by taking $\tilde{U}$ to be $\partial/\partial x^\mu$ are consequences of these, and simply express the partial derivatives of
$H\circ\chi$ as \lq total\rq\  derivatives of $H$.}:

\begin{equation}\label{hamilton_equations}
	\frac{\partial u^a}{\partial x^\mu} = \frac{\partial H}{\partial \rho^\mu_a}\, ; \qquad
	\frac{\partial \rho^\mu_a}{\partial x^\mu} = -\frac{\partial H}{\partial u^a} \, ,
\end{equation}
where a summation on $\mu$ is understood in the last equation.   Note that had we not changed to normal coordinates on $M$, the volume form would not have the above simple form and therefore there would be related extra terms in the previous expressions and in equations (2.15).

These Hamilton equations are often described as being covariant. This term
must be treated with caution in this context. Clearly, by writing the equations
in the invariant form $\chi^*(i_{\tilde{U}}\d \Theta _H)=0$ we have shown that they are in
a sense covariant. However, it is important to remember that the function $H$
is, in general, only locally defined; in other words, there is in general no true
`Hamiltonian function', and the local representative $H$ transforms in a
non-trivial way under coordinate transformations. When $M(E)$ is a trivial
bundle over $P(E)$, so that there is a predetermined global section, then
the Hamiltonian section may be represented by a global function and no problem
arises. This occurs for instance when $E$ is trivial over $M$.
In general, however, there is no preferred section of $M(E)$ over
$P(E)$ to relate the Hamiltonian section to, and in order to write the
Hamilton equations in manifestly covariant form one must introduce a
connection. (See [Ca91] for a more detailed discussion.) 

%%%%%%%%%%%%%%%%%%%%%

\subsection{The fundamental formula}\label{sec:fundamental}

Thus we have obtained the formula that relates the differential of the action with a 1-form on a space of fields on the bulk manifold and a 1-form on a space of fields at the boundary.
\begin{equation}\label{fundamental}
\mathrm{d} S_\chi = \mathrm{EL}_\chi +  \Pi^* \alpha_\chi \, , \qquad \chi \in \mathcal{F}_{P(E)} \, . 
\end{equation}
In the previous equation $\mathrm{EL}_\chi$ denotes the Euler-Lagrange 1-form on the space of fields $\chi = (\Phi, P)$ with local expression (using variational derivatives):
\begin{equation}\label{ELform}
\mathrm{EL}_\chi = \left(  \frac{\partial \Phi^a}{\partial x^\mu}  - \frac{\partial H}{\partial P^\mu_a}  \right) \delta P_a^\mu  - \left( 	\frac{\partial P^\mu_a}{\partial x^\mu} + \frac{\partial H}{\partial \Phi^a}  \right) \delta \Phi^a \, ,
\end{equation}
or, more explicitly:
$$
\mathrm{EL}_\chi (\delta \Phi, \delta P) = \int_M \left[ \left(  \frac{\partial \Phi^a}{\partial x^\mu}  - \frac{\partial H}{\partial P^\mu_a}  \right) \delta P_a^\mu  - \left( \frac{\partial P^\mu_a}{\partial x^\mu} + \frac{\partial H}{\partial \Phi^a}  \right) \delta \Phi^a \right] \, \mathrm{vol}_M \, .
$$

In what follows we will denote by $(P(E), \Theta_H)$ the covariant Hamiltonian field theory with bundle structure $\pi \colon E \to M$ defined over the $m$-dimensional manifold with boundary $M$, Hamiltonian function $H$ and canonical $m$-form $\Theta_H$.    

We will say that the action $S$ is regular if the set of solutions of Euler-Lagrange equations $\mathcal{EL}_M$ is a submanifold of $\mathcal{F}_{P(E)}$. 
Thus we will also assume when needed that the action $S$ is regular (even though this must be proved case by case) and that the projection $\Pi(\mathcal{EL})$ to the space of fields at the boundary $T^*\mathcal{F}_{\partial M}$ is a smooth manifold too.

%%%%%%%%%%%%%%%%%%%%%%%%%%

\section{The presymplectic formalism at the boundary}\label{sec:presymplectic}

%%%%%%%%%%%%%%%%%%%%%%%%
%%%%%%%%%%%%%%%%%%%%%%%%

\subsection{The evolution picture near the boundary}\label{sec:dynamical_eqs}

  We discuss in what follows the evolution picture of the system near the boundary.  As discussed in Section \ref{sec:cotangent_boundary}, we assume that there exists a collar $U_\epsilon \cong (-\epsilon , 0]\times\partial M$ of the boundary \ $\partial M$ with adapted coordinates $(t; x^1,\ldots, x^d)$, where $t = x^0$ and where $x^i$, $i = 1,\ldots, x^d$ define a local chart in $\partial M$. The normal coordinate $t$ can be used as an evolution parameter in the collar. We assume again that the volume form in the collar is of the form $\mathrm{vol}_{U_{\epsilon}} = dt\wedge \mathrm{vol}_{\partial M}$. 
  
 If $M$ happens to be a globally hyperbolic space-time $M \cong [t_0,t_1]\times \Sigma$ where $\Sigma$ is a Cauchy surface, $[t_0,t_1] \subset \mathbb{R}$ denotes a finite interval in the real line, and the metric has the form $-dt^2 + g_{\partial M}$ where $g_{\partial M}$ is a fixed Riemannian metric on $\partial M$, then $t$ represents a time evolution parameter throughout the manifold and the volume element has the form $\mathrm{vol}_M = dt \wedge \mathrm{vol}_{\partial M}$. Here, however, all we need to assume is that our manifold has a collar at the boundary as described above.

Restricting the action $S$ of the theory to fields defined on $U_\epsilon$, i.e., sections of the pull-back of the bundles $E$ and $P(E)$ to $U_\epsilon$, we obtain,
\begin{equation}\label{Sepsilon}
S_\epsilon (\chi ) = \int_{U_\epsilon} \chi^*\Theta_H = \int_{-\epsilon}^0 \d t \int_{\partial M} \mathrm{vol}_{\partial M} \left[ P_a^0 \partial_0 \Phi^a + P_a^k\partial_k \Phi^a - H(\Phi^a, P_a^0, P_a^k) \right]\, . 
\end{equation}
Defining the fields at the boundary as discussed in Lemma \ref{decomposition},
$$
\varphi^a = \Phi^a|_{\partial M} \, , \qquad p_a = P_a^0|_{\partial M} \, , \qquad \beta_a^{k} = P_a^{k}|_{\partial M} \, ,
$$
we can rewrite \eqref{Sepsilon} as
$$
S_\epsilon (\chi) = \int_{-\epsilon}^0 \d t \int_{\partial M} \mathrm{vol}_{\partial M}[ p_a \dot{\varphi}^a + \beta_a^k\partial_k\varphi^a -H(\varphi^a,p_a,\beta_a^k)] \, .
$$
Letting $\langle p, \dot{\varphi} \rangle = \int_{\partial M} p_a \dot{\varphi}^a \,\,  \mathrm{vol}_{\partial M}$ denote, as in \eqref{pairing_cotangent}, the natural pairing and, similarly, 
$$\langle \beta, \mathrm{d}_{\partial M}\varphi \rangle = \int_{\partial M} \beta_a^k \partial_k \varphi^a \,  \mathrm{vol}_{\partial M},
$$
we can define a density function $\mathcal{L}$ as,
\begin{equation}\label{densityL}
\mathcal{L}(\varphi,\dot{\varphi},p,\dot{p},\beta,\dot{\beta})=\langle p,\dot{\varphi}\rangle + \langle\beta, \mathrm{d}_{\partial M}\varphi \rangle - \int_{\partial M} H(\varphi^a,p_a,\beta_a^k) \, \mathrm{vol}_{\partial M} \, ,
\end{equation}
and then
$$
S_\epsilon (\chi) = \int_{-\epsilon}^0 \d t \, \, {\mathcal{L}}(\varphi,\dot{\varphi},p,\dot{p},\beta,\dot{\beta}) \, .
$$

Notice again that because of the existence of the collar $U_\epsilon$ near the boundary and the assumed form of  $\mathrm{vol}_ {U_\epsilon}$,  the elements in the bundle $i^*P(E)$ have the form $\rho_a^0 \d u^a \wedge \mathrm{vol}_{\partial M} + \rho_a^k \d u^a \wedge \d t \wedge i_{\partial /\partial x^k}\mathrm{vol}_{\partial M}$ and, as discussed in Lemma \ref{decomposition}, the bundle $i^*P(E)$ over $i^*E$ is isomorphic to the product $\bigwedge_1^m(i^*E) \times B$, where $B = \bigwedge_1^{m-1}(i^*E)$.   The space of double sections $(\varphi,p)$ of the bundle $\bigwedge_1^m(i^*E) \to i^*E \to \partial M$ correspond to the cotangent bundle $T^*\mathcal{F}_{\partial M}$ and the double sections $(\varphi, \beta)$ of the bundle $B \to i^*E \to \partial M$ correspond to a new space of fields at the boundary denoted by $\mathcal{B}$.

We will introduce now the total space of fields at the boundary $\mathcal{M}$ which is the space of double sections of the iterated bundle $i^*P(E) \to i^*E \to \partial M$. Following the previous remarks it is obvious that $\mathcal{M}$ has the form,
$$\mathcal{M}= \mathcal{T}^*\mathcal{F}_{\partial M} \times_{\mathcal{F}_{\partial M}} \mathcal{B} = \{(\varphi, p, \beta)\} \, .$$

Thus the density function $\mathcal{L}$, Eq. \eqref{densityL}, is defined on the tangent space $T\mathcal{M}$ to the total space of fields at the boundary and could be called accordingly the boundary Lagrangian of the theory.

Consider the action $ A = \int_{-\epsilon}^0 \mathcal{L} \,\, \d t$ defined on the space of curves $\sigma\colon ( -\epsilon, 0] \to \mathcal{M}$.
If we compute $\mathrm{d}A$ we obtain a bulk term, that is, an integral on $(-\epsilon, 0]$, and a term evaluted at $\partial [-\epsilon,0] = \{-\epsilon, 0\}$. Setting the bulk term equal to zero, we obtain the Euler-Lagrange equations of this system considered as a Lagrangian system on the space $\mathcal{M}$ with Lagrangian function $\mathcal{L}$, 
\begin{equation}\label{Euler-Lagrange Equations 1}
 \frac{\d}{\d t} \frac{{\delta}{\mathcal{L}}}{{\delta}{\dot{{\varphi}^a}}}= \frac{{\delta}{\mathcal{L}}}{{\delta}{{\varphi}^a}} \, ,
 \end{equation}
 which becomes,
 \begin{equation}\label{pidot}
\dot{p_a}= -{\partial}_k{\beta}_a^k - \frac{{\partial}H}{{\partial}{\varphi}^a} \, .
\end{equation}
Similarly, we get for the fields $p$ and $\beta$:
$$
%\begin{equation}%\label{Euler-Lagrange Equation 2}
\frac{\d}{\d t}\frac{{\delta}{\mathcal{L}}}{{\delta}{\dot{p}_a}}=\frac{{\delta}{\mathcal{L}}}{{\delta}{p_a}} \, , \quad
 \frac{\d}{\d t}\frac{{\delta}{\mathcal{L}}}{{\delta}{\dot{{\beta}}_a^k}}=\frac{{\delta}{\mathcal{L}}}{{\delta}{{\beta}_a^k}}%\end{equation}
 $$
that become respectively,
\begin{equation}\label{phidot}
\dot{\varphi}_a = \frac{\partial H}{\partial p_a} \, ,
\end{equation}
and, the constraint equation:
\begin{equation}\label{dconstraint}
\d_{\partial M}{\varphi}-\frac{\partial H}{{\partial}{{\beta}_a^k}}=0 \, .
\end{equation}
Thus, Euler-Lagrange equations in a collar $U_\epsilon$ near the boundary, can be understood as a system of evolution equations on $T^*\mathcal{F}_{\partial M}$ depending on the variables $\beta_a^k$, together with a constraint condition on the extended space $\mathcal{M}$.  The analysis of these equations, Eqs. (3.4), (3.5) and (3.6), is best understood in a presymplectic framework.

\subsection{The presymplectic picture at the boundary and constraints analysis}\

We will introduce now a presymplectic framework on $\mathcal{M}$ that will be helpful in the study of Eqs.(3.4)-(3.6).

Let $\varrho :\mathcal{M} \longrightarrow \mathcal{T}^*\mathcal{F}_{\partial M}$ denote the canonical projection $\varrho(\varphi,p,\beta)=(\varphi,p)$.  (See Figure 2.)
Let $\Omega$ denote the pull-back of the canonical symplectic form ${\omega}_{\partial M}$ on $\mathcal{T}^*\mathcal{F}_{\partial M}$ to $\mathcal{M}$, i.e., let $\Omega=\varrho^*\omega_{\partial M}$.   Note that the form $\Omega$ is closed but degenerate, that is, it defines a presymplectic structure on $\mathcal{M}$. An easy computation shows that the characteristic distribution $\mathcal{K}$ of $\Omega$, is given by
$$
\mathcal{K} = \ker\Omega= \mathrm{span} \left\{ \frac{\delta}{{\delta}{\beta}_a^k} \right\} \, .
$$

Let us consider the function defined on $\mathcal{M}$, 
$$
\mathcal{H}(\varphi,p,\beta)= -\langle \beta, d_{\partial M}\varphi \rangle + \int_{\partial M} H({\varphi}^a, p_a, {\beta}_a^k)\, \mathrm{vol}_{\partial M} \, .
$$

\begin{figure}[ht]
\centering
\includegraphics[width=10cm]{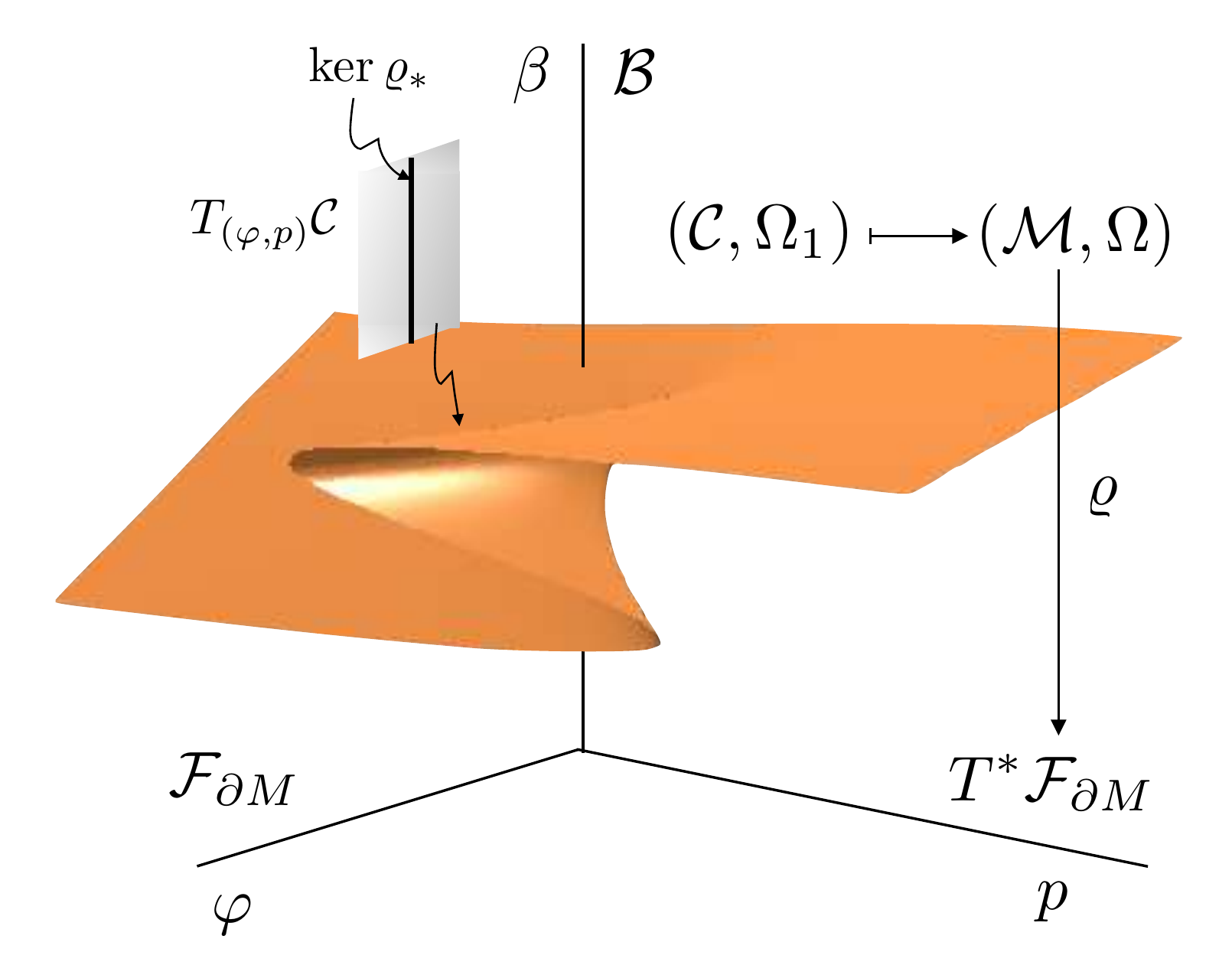}
\caption{The space of fields at the boundary $\mathcal{M}$ and its relevant structures.}\label{diagram}
\end{figure}

We will refer to $\mathcal{H}$ as the boundary Hamiltonian of the theory.
 Thus $\mathcal{L}$ can be rewritten as
$$\mathcal{L}(\varphi, \dot{\varphi},p,\dot{p},\beta, \dot{\beta})=\langle p,\dot{\varphi} \rangle - \mathcal{H}(\varphi,p,\beta)$$ and 
\begin{equation}\label{Lagrangian with Boundary Hamiltonian}
 S_{\epsilon}(\varphi, p, \beta) =\int_{-\epsilon}^0 [\langle p,\dot{\varphi} \rangle - \mathcal{H}(\varphi,p,\beta)] d t \, ,
 \end{equation}
and therefore the Euler-Lagrange equations (3.8) and (3.9) can be written as
\begin{equation}\label{Hamilton's Evolution Equations}
 \dot{\varphi}^a = \frac{\delta \mathcal{H}}{\delta p_a}\, ,\qquad 
\dot{p}_a = -\frac{\delta \mathcal{H}}{\delta{\varphi}^a}  \, ,
\end{equation}
and
\begin{equation}\label{Hamilton's Constraint Equation}
 0= \frac{\delta\mathcal{H}}{\delta{\beta}_a^k} \, .
\end{equation}

Now it is easy to prove the following:

\begin{theorem}\label{presymplectic_equation}
The solutions to the equations of motion defined by the Lagrangian $\mathcal{L}$ over a collar $U_\epsilon$ at the boundary, $\epsilon$ small enough, are in one-to-one correspondence with the integral curves of the presymplectic system $(\mathcal{M},\Omega,\mathcal{H})$, i.e., with the integral curves of the vector field $\Gamma$ on $\mathcal{M}$ satisfying 
\begin{equation}\label{presymplectic_equation1}
i_\Gamma \Omega = \d\mathcal{H} \, .
\end{equation}
\end{theorem}

\begin{proof}  Let $\Gamma = A^a\frac{\delta}{\delta{\varphi}^a} + B^a\frac{\delta}{{\delta}p^a} + C^a\frac{\delta}{{\delta}{\beta}_a^k}$ be a vector field on $\mathcal{M}$ (notice that we are using an extension of the functional derivative notation introduced in Section \ref{sec:cotangent_boundary} on the space of fields $\mathcal{M}$).  Then because $\Omega = \delta{\varphi}^a \wedge \delta p_a$, we get from $i_{\Gamma}\Omega= \d\mathcal{H}$ that,
$$
A^a = \frac{{\delta}{\mathcal{H}}}{{\delta} p_a},\qquad
B^a = -\frac{{\delta}{\mathcal{H}}}{{\delta}{\varphi}^a} \, ,\qquad 0 = \frac{{\delta}{\mathcal{H}}}{{\delta}{\beta}_a^k} \, .
$$

Thus, $\Gamma$ satisfies Eq. \eqref{presymplectic_equation1} iff 
$$
\dot{\varphi}^a =\frac{{\delta}{\mathcal{H}}}{{\delta}p_a}, \qquad
\dot{p}_a = -\frac{{\delta}{\mathcal{H}}}{{\delta}{\varphi}^a} \, , \quad 
\mathrm{and} \quad 
0= \frac{{\delta}{\mathcal{H}}}{{\delta}{\beta}_a^k} \, .
$$
\end{proof}

Let us denote by $\mathcal{C}$ the submanifold of the space of fields $\mathcal{M} =T^*\mathcal{F}_{\partial M} \times \mathcal{B}$ defined by eq. $(3.9)$.  It is clear that the
restriction of the solutions of the Euler-Lagrange equations on $M$ to the boundary ${\partial}M$, are contained in $\mathcal{C}$; i.e., $ \Pi (\mathcal{EL})\subset\mathcal{C}.$

Given initial data $\varphi, p$ and fixing $\beta$, existence and uniqueness theorems for initial value problems when applied to the initial value problem above, would show the existence of solutions for small intervals of time, i.e., in a collar near the boundary.   

However, the constraint condition given by eq. \eqref{Hamilton's Constraint Equation}, satisfied automatically by critical points of $S_\epsilon$ on $U_\epsilon$, must be satisfied along the integral curves of the system, that is, for all $t$ in the neighborhood $U_\epsilon$ of $\partial M$.   This implies that consistency conditions on the evolution  must be imposed.  Such consistency conditions are just that the constraint condition eq. \eqref{Hamilton's Constraint Equation}, is preserved under the evolution defined by eqs. \eqref{Hamilton's Evolution Equations}. This is the typical situation that we will find in the analysis of dynamical problems with constraints and that we are going to summarily analyze in what follows.

%%%%%%%%%%%%%%%%%%%%
%%%%%%%%%%%%%%%%%%%%

\subsubsection{The Presymplectic Constraints Algorithm (PCA)}

Let $i$ denote the canonical immersion $\mathcal{C}=\{(\varphi,p,\beta)| \frac{\delta {\mathcal{H}}}{\delta{\beta}}=0\}\to \mathcal{M}$ and consider the pull-back of $\Omega$
to $\mathcal{C}$, i.e., $\Omega_1 = i^*\Omega$. Clearly then, $\ker \Omega_1 = \ker \varrho_* \cap T\mathcal{C}$.
But $\mathcal{C}$ is defined as the zeros of the function $\delta \mathcal{H}/\delta \beta$. Therefore if $\delta^2 \mathcal{H}/\delta^2\beta$ is nondegenerate (notice that the operator $\delta^2 \mathcal{H}/\delta
{\beta}_a^i{\delta}{\beta}_b^j$ becomes the matrix $\partial^2 H /\partial \beta_a^i \partial\beta_b^j$), by an appropriate extension of the Implicit Function Theorem, we could solve $\beta$ as a function of $\varphi$ and $p$.  In such case, locally, $\mathcal{C}$ would be the graph of a function $F\colon T^*\mathcal{F}_{\partial M} \to \mathcal{B}$, say $\beta = F(\varphi, p)$. 
Collecting the above yields:

\begin{proposition}  The submanifold $(\mathcal{C}, \Omega_1)$ of $(\mathcal{M},\Omega,\mathcal{H})$ is symplectic iff $H$ is regular, i.e., 
$\partial^2 H /\partial \beta_a^i \partial\beta_b^j$ is non-degenerate.  In such case the projection $\varrho$ restricted to $\mathcal{C}$, which we denote by $\varrho_C$, is a local symplectic diffeomorphism and therefore $\varrho_C^*\omega_{\partial M} = \Omega_1$.
\end{proposition}

When the situation is not as described above,
and $\beta$ is not a function of $\varphi$ and $p$, then $(\mathcal{C},\Omega_1)$ is indeed a presymplectic submanifold of $\mathcal{M}$ and $i_{\Gamma}{\Omega}=d\mathcal{H}$ will not hold necessarily at every point in $\mathcal{C}$.  In this case 
 we would apply Gotay's Presymplectic Constraints Algorithm [Go78], to obtain  the maximal submanifold of $\mathcal{C}$ for which 
 $i_{\Gamma}{\Omega}=d\mathcal{H}$ is consistent and that can be summarized as follows.

%%%%%%%%%%%%%%%%

Consider a presymplectic system $(\mathcal{M}, \Omega, \mathcal{H})$ where $\mathcal{M}= T^*\mathcal{F}_{\partial M}\times\mathcal{B}$ and, $\Omega$ and $\mathcal{H}$ are as defined above. Let $\mathcal{M}_0 = \mathcal{M}$, $\Omega_0 = \Omega$, $\mathcal{K}_0 = \ker \Omega_0$, and $\mathcal{H}_0 = \mathcal{H}$.  We define the primary constraint submanifold $\mathcal{M}_1$ as the submanifold defined by the consistency condition for the equation $i_\Gamma \Omega_0 = \d\mathcal{H}_0$, i.e., 
$$
\mathcal{M}_1  = \{ \chi \in \mathcal{M}_0 \mid \langle Z_0(\chi) , \d\mathcal{H}_0(\chi) \rangle = 0, \, \,  \forall Z_0 \in \mathcal{K}_0 \} \, .
$$
 Thus $\mathcal{M}_1= \mathcal{C}$. Denote by $i_1 \colon \mathcal{M}_1 \to \mathcal{M}_0$ the canonical immersion. Let $\Omega_1 = i_1^*\Omega_0$, $\mathcal{K}_1 = \ker \Omega_1$, and $\mathcal{H}_1 = i_1^*\mathcal{H}_0$.   We now define recursively the $(k+1)$-th constraint submanifold as the consistency condition for the equation $i_\Gamma \Omega_k = \d\mathcal{H}_k$, that is,
$$
\mathcal{M}_{k+1}  = \{ \chi \in \mathcal{M}_k \mid \langle Z_k(\chi) , \d\mathcal{H}_k(\chi) \rangle = 0, \, \,  \forall Z_k \in \mathcal{K}_k \} \, \qquad k \geq 1 \, ,
$$
and $i_{k+1}\colon \mathcal{M}_{k+1} \to \mathcal{M}_k$ is the canonical embbeding (assuming that $\mathcal{M}_{l+1}$ is a regular submanifold of $\mathcal{M}_k$), and  $\Omega_{k+1} = i_{k+1}^*\Omega_k$, $\mathcal{K}_{k+1} = \ker \Omega_{k+1}$ and $\mathcal{H}_{k+1} = i_{k+1}^*\mathcal{H}_k$.

The algorithm stabilizes if there is an integer $r> 0$ such that 
$\mathcal{M}_{r} = \mathcal{M}_{r+1}$.
We refer to this $\mathcal{M}_r$ as the final constraints submanifold and we denote it by $\mathcal{M}_\infty$. Letting $i_\infty\colon \mathcal{M}_\infty \to \mathcal{M}_0$ denote the canonical immersion, we define,
$$
\Omega_\infty = i_\infty^*\Omega_0, \qquad \mathcal{K}_\infty = \ker \Omega_\infty\, , \qquad \mathcal{H}_\infty = i_\infty^*\mathcal{H}_0 \, .
$$
Notice that the presymplectic system $(\mathcal{M}_\infty, \Omega_\infty,  \mathcal{H}_\infty )$ is always consistent, that is, the dynamical equations defined by $i_\Gamma \Omega_\infty = d\mathcal{H}_\infty$ will always have solutions on $\mathcal{M}_\infty$.  The solutions will not be unique if $\mathcal{K}_\infty \neq 0$, hence the integrable distribution $\mathcal{K}_\infty$ will be called the ``gauge'' distribution of the system, and its sections (that will necessarily close a Lie algebra), the ``gauge'' algebra of the system.\\

The quotient space $\mathcal{R} = \mathcal{M}_\infty / \mathcal{K}_\infty$, provided it is a smooth manifold,  inherits a canonical symplectic structure $\omega_\infty$ such that $\pi_\infty^*\omega_\infty = \Omega_\infty$, where $\pi_\infty \colon \mathcal{M}_\infty \to \mathcal{R}$ is the canonical projection. We will refer to it as the reduced phase space of the theory.   Notice that the Hamiltonian $\mathcal{H}_\infty$ also passes to the quotient and we will denote its projection by $h_\infty$ i.e., $\pi_\infty^* h_\infty = \mathcal{H}_\infty$.  

\begin{theorem}\label{reduction_theorem}
The reduction $\widetilde{\Pi}(\mathcal{EL})$  of the submanifold of Euler-Lagrange fields of the theory is an isotropic submanifold of the reduced phase space $\mathcal{R}$ of the theory.
\end{theorem}

\begin{proof}      
Recall that $\Pi(\mathcal{EL}) \subset \mathcal{C}$. It is clear that $\Pi(\mathcal{EL}) \subset \Pi(\mathcal{EL}_\epsilon) \subset \mathcal{M}_{\infty}$ where $\mathcal{EL}_\epsilon = \mathcal{EL}_{U_\epsilon}$ are 
the critical points of the action $S_\epsilon$, i.e., solutions of the Euler-Lagrange equations of the theory on $U_\epsilon$.

 The reduction $\widetilde{\Pi}(\mathcal{EL}) = \Pi(\mathcal{EL})/ (\mathcal{K}_\infty\cap T\, \Pi (\mathcal{EL}))$ of the isotropic submanifold  $ \Pi(\mathcal{EL})$ to the reduced phase space $\mathcal{R} = \mathcal{M}_\infty / \mathcal{K}_\infty$ is isotropic because $\pi_\infty^*\omega_\infty = \Omega_\infty$, hence $\pi_\infty^* (\omega_\infty\mid_{\widetilde{\Pi}(\mathcal{EL})}) = (\pi_\infty^* \omega_\infty)\mid_{\Pi(\mathcal{EL})} = \varrho^*\mathrm{d} \alpha \mid_{\Pi(\mathcal{EL})} = 0$. 
\end{proof}

%%%%%%%%%%%%%%%%%%%%%
%%%%%%%%%%%%%%%%%%%%%

\subsection{The limit $\lambda \to 0$ of Yang-Mills theories}
Recall equations $(2.5)$ and $(2.6)$ for the action of Yang-Mills theories in a first-order Hamiltonian formulation of the theory:

\begin{equation}\label{ymap}
S_{\mathrm{YM},{\lambda}}(A,P) = \int_M P_a^{\mu\nu} dA_\mu^a \wedge \d x^{m-1}_\nu - H_{\lambda}(A,P) \mathrm{vol}_M \, .
\end{equation}
with Hamiltonian function,
\begin{equation}\label{hamiltonian}
H_{\lambda}(A,P) = \frac{1}{2} \epsilon_{bc}^aP^{\mu\nu}_a A_\mu^bA_\nu^c + \frac{\lambda}{4}P^{\mu\nu}_a P_{\mu\nu}^a \, 
\end{equation}
for some ${\lambda} \geq 0$,
where the indexes $\mu\nu$  ($a$) in $P_a^{\mu\nu}$ have been lowered (raised) with the aid of the Lorentzian metric $\eta$ (the Killing-Cartan form on $\mathfrak{g}$, respect.). \\

 Plugging $(2.19)$ into $(2.18)$ and expanding the right hand side of $(2.18)$, we obtain,
 \begin{equation}\label{ymP}
S_{\mathrm{YM},{\lambda}}(A,P) = -\int_M \frac{1}{2} \left[ P^{\mu\nu}_a (\partial_\mu A_\nu^a - \partial_\nu A_\mu^a + \epsilon_{bc}^a A_\mu^b A_\nu^c) +  \frac{{\lambda}}{2}P^{\mu\nu}_aP^a_{\mu\nu} \right] \, \mathrm{vol}_M \, .
\end{equation}
Using that the curvature,
\begin{eqnarray}\label{Fmunu}
F_A &=& d_A A = dA + \frac{1}{2}[A\wedge A] = F_{\mu\nu} \d x^\mu \wedge \d x^\nu \\
&=& \frac{1}{2}\left( \partial_\mu A_\nu^a - \partial_\nu A_\mu^a + \epsilon_{bc}^a A_\mu^b A_\nu^c\right) \d x^\mu \wedge \d x^\nu \otimes \xi_a\,  \nonumber
\end{eqnarray}
we can rewrite eqn $(2.20)$ as
$$
S_{\mathrm{YM},{\lambda}}(A,P) = -\int_M \left[ P_a^{\mu\nu} F_{\mu\nu}^a + \frac{\lambda}{4}P^{\mu\nu}_aP^a_{\mu\nu} \right] \, \mathrm{vol}_M \, . 
$$

This last expression is the action of the Yang-Mills theory for any given $\lambda \geq 0$.
If we take its limit 
  $\lambda \to 0$, we obtain,
\begin{equation}\label{YM,0action}
S_{YM,0} (A,P) = \int_M P^{\mu\nu}_a F_{\mu\nu}^a \mathrm{vol}_M \,
\, ,
\end{equation}
whose equations of motion are given by,
$$
F_A = 0\, , \qquad d_A^*P = 0 \, .
$$
Thus the moduli space of solutions of the Euler-Lagrange equations is given by,
$$
\mathcal{M} = \{ F_A = 0, d_A^*P = 0 \} /\mathcal{G}_M \, ,
$$
where $\mathcal{G}_M$ denotes the group of gauge transformations of the theory.

%%%%%%%%%%%%%%%%%%%%%%%%%%

\section{Palatini's Gravity}

\subsection{Palatini's Yang-Mills}

The primary fields of a theory of gravity a la Palatini will be given by principal connections $A$
on a $G$-principal bundle over a smooth manifold $M$ with $G$ the Lorentz group $O(1,d)$, the group of isometries preserving the non-degenerate quadratic form $Q$ of signature $-+\cdots +$, with $m = 1+d = \dim M$.

The connections $A$ can be considered as vertical equivariant 1-forms on a principal fiber bundle with structural group $O(1,d)$. The choice of the principal fiber bundle $P \to M$ determines a sector of a full theory of gravity where, in addition to the bundle $P$, we should consider all equivalence classes of principal $O(1,d)$ bundles over $M$.    If we fix a topology on $M$, the corresponding family of classes of principal fiber bundles are in one-to-one correspondence with homotopy classes of maps $f \colon M \to B_{O(1.d)}$, where $B_{O(1.d)}$ is the universal classifying space of the Lorentz group and the principal fiber bundle corresponding to the map $f$ is given by $P_f = f^*E_{O(1,d)}$, where $E_{O(1.d)} \to B_{O(1.d)}$ is the universal principal $O(1,d)$ bundle.  Thus the fields corresponding to each equivalence class will define connected components in the space of all fields and we will focus on one of them.

\subsection{Palatini's constraint}

Palatini's constraint determines a subbundle of the covariant phase space whose sections define a submanifold of the space of fields $J^1\mathcal{F}^*$ such that the restriction of the topological sector of $SO(1,3)$--Yang-Mills is equivalent to Palatini's action.   

Consider the bundle $F= GL(\tau_m,TM) \subset \mathrm{Hom}(\tau_m,TM) \cong \tau_m^* \otimes TM$ over $M$ whose fiber at $x\in M$ consists on invertible linear maps $e(x)$ from $\tau_m(x)$ to $T_xM$ and where $\tau_m = M\times \mathbb{M}^m$ is the trivial bundle over $M$ with fiber the $m$-dimensional Minkowski space $\mathbb{M}^m$ with metric $\eta = \mathrm{diag}(-, + \cdots , +)$.  Notice that local cross sections of the bundle $F$ can be thought as local frames on $M$, i.e., if $U$ is an open set on $M$ such that $TM\mid_U \cong U \times \mathbb{R}^m$, then a cross section $e \colon U \to \tau_m^* \otimes TM$, defines a map $e_x := e(x) \colon \mathbb{R}^m \to T_xM$ for each $x \in U$, i.e.,  a family of linearly independent vectors $e_I (x)$, $I = 0,1,\ldots, d$, which are the images under $e_x$ of the standard orthogonal basis $u_i$ on $\mathbb{M}^m$, that is $\eta(u_0,u_0) = -1$, $\eta (u_k, u_k) = 1$, $k = 1, \ldots, d$.   With an slight abuse of notation we will denote $e_x(u_I) = e_I(x)$.    Global cross sections $e$ are usually called vierbeins for an arbitrary dimension $m$, or tetrad fields if $m = 4$.   In what follows we will not assume that there are globally defined sections of $F$ (that it may not exist).    Notice that given a local cross section $e$ it defines a Lorentz metric on $U$ by means of $g_x (u,v) = \eta_x(e^{-1}(u), e^{-1}(v))$ for any $u,v\in T_xU$.   The metric $g$ is Lorentz because clearly the vectors $e_I(x)$ determine an orhonormal basis for $g$ at $T_xM$ such that $g_x(e_I(x), e_J(x))$ is diagonal with diagonal $(-,+\ldots,+)$. 

Choosing local coordinates $x^\mu$ on $U$ we will have that $e_I = e_I^\mu (x) \partial /\partial x^\mu$ will defined a local vector field on $U$ for each $I$.   With this notation we may also write the local cross section $e$ as $e = e_I \otimes u^I = e_I^\mu (x) \partial/\partial x^\mu \otimes u^I$ where $u^I$ denotes the canonical dual basis of the standard orthogonal basis $u_I$.  

Let us recall that we have a distinguished volume form $\mathrm{vol}_M$ on $M$, i.e., a global section of the determinant bundle $\mathrm{det}(M) = \Lambda^m(TM)$.  Morevoer there is a canonical section of the bundle $\mathrm{det}(\tau_m) = \Lambda^m(\tau_m)$ given by $\mathrm{vol}_{\eta} = u^0 \wedge u^1 \wedge \cdots \wedge u^d$. Then a linear map $e_x\colon \tau_m(x) \to T_xM$ defines a pull-back $e^*(\mathrm{vol}_M) = \epsilon \mathrm{vol}_{\eta}$, in other words, $\epsilon (x)$ is the determinant of the map $e_x$.  In local coordinates:
$$
\epsilon (x) = \det (e_I^\mu(x)) \, .
$$

Consider the map $P \colon F \to P(E)$ defined as:
$$
P(e) = \epsilon e\wedge e
$$
where $e\wedge e$ is defined as the linear map from $\tau_m\wedge \tau_m$ to $T_xM \wedge T_xM$ given by $e\wedge e(u\wedge v) = e(u) \wedge e(v)$.    Using the previous notation we may write:
$$
P(e) = \epsilon e_I^\mu e_J^\nu \frac{\partial}{\partial x^\mu}\wedge \frac{\partial}{\partial x^\nu} \otimes u^I \wedge u^J \, .
$$
Notice that if we write the tensor $P(e)$ in the local basis $\frac{\partial}{\partial x^\mu}\wedge \frac{\partial}{\partial x^\nu} \otimes u^I \wedge u^J$ as:
$$
P (e) = P^{\mu\nu}_{IJ} \frac{\partial}{\partial x^\mu}\wedge \frac{\partial}{\partial x^\nu} \otimes u^I \wedge u^J \, ,
$$
then 
$$
P_{IJ}^{\mu\nu} = \det(e_I^\mu) \, e_{[I}^{[\mu} e_{J]}^{\nu]} \ ,
$$
with 
$P^{\mu\nu}_{IJ} = - P^{\nu\mu}_{IJ} = - P^{\mu\nu}_{JI} = P^{\nu\mu}_{JI}$.  We will sometimes use the notation 
$P^{\mu\nu}_{IJ} = \det(e_I^\mu) \, e_{I}^{\mu} \wedge e_{J}^{\nu}$ to indicate the skew symmetry in the pairs of indices $IJ$ and $\mu\nu$.

Finally notice that $P(e)$ actually lies in $P(E)$ as the fiber of $P(E)$ at $x$ is given by $T_xM\wedge T_xM\otimes \mathfrak{so}(1,d)$ and $\tau_m\wedge \tau_m \subset \mathfrak{so}(1,d)$.  

The image of $F$ under the map $P$ will be called the Palatini subbundle of $P(E)$ and will be denoted simply by $P(F) \subset P(E)$.   Double sections of this bundle are the fields of the theory we are interested in. Such space of sections will be denoted as $\mathcal{P} \subset J^1\mathcal{F}_M^*$.  Notice that a double section $(A, P)$ of $\mathcal{P}$ is a section of $P(E)$ such that locally there exists $e$ such that $P = \epsilon \, e\wedge e$.

Hence the space of fields of the theory we are constructing can be considered as a submanifold of the space of fields $J^1\mathcal{F}_M^*$ defined by the range of the map $P$.

%%%%%%%%%%%%%%%%%%%%%%%%%%%
%%%%%%%%%%%%%%%%%%%%%%%%%%%

\subsection{The action}

The action of the topological phase of Yang-Mills given by Eq. $(2.22)$ is given by:
$$
S_{YM,0} = \int_M P_{IJ}^{\mu\nu} F_{\mu\nu}^{IJ} \mathrm{vol}_M \, ,
$$
with $(A,P)\in J^1\mathcal{F}_M^*$, then if we restrict $(A,P)$ to $\mathcal{P}$, the
action becomes:
$$
S_{YM,0}\mid_{\mathcal{P}} = \int_M \epsilon \, e_I^\mu e_J^\nu F_{\mu\nu}^{IJ} \mathrm{vol}_M \, ,
$$ 
which is exactly Palatini's action for gravity.

The Euler-Lagrange equations of the theory can be obtained by standard methods by computing the differential of $S_{YM,0}$ restricted to $\mathcal{P}$ or, alternatively, using an appropriate version of Lagrange's multipliers theorem to obtain the critical points of $S_{YM,0}$ restricted to $\mathcal{P}$.   We will develop this point of view in the following section.

%%%%%%%%%%%%%%%%%%%%%%%%%%%

\subsection{Critical points and Euler-Lagrange equations}

\subsubsection{Lagrange's multipliers theorem}

We will discuss first the version of Lagrange's multipliers theorem suited to the problem at hand.

Theorem:
 Let $\mathcal{M}$ be an affine manifold and let  $F\colon \mathcal{M} \to \mathbb{R}$ be a differentiable function.  Let $\mathcal{D}$ be a smooth manifold and let $\Phi \colon \mathcal{D} \to \mathcal{M}$ be a smooth injective function. Let $\mathcal{N} = \{ x \in \mathcal{M} \mid \exists e \in \mathcal{D}\, , x = \Phi (e) \}$.
 
   $x\in \mathcal{N}$ is a critical point of $F\mid_{\mathcal{N}} \colon \mathcal{N} \to \mathbb{R}$  iff there exists $e \in \mathcal{D}$ and $\lambda \in \mathcal{M^*}$ such that $(x,\lambda,e)$ is a critical point of the extended function $\mathbb{F} \colon \mathcal{M} \times \mathcal{M}^* \times \mathcal{D} \to \mathbb{R}$ given by:
$$
\mathbb{F}(x,\lambda, e) = F(x) + \langle \lambda, x - \Phi (e)\rangle \ .
$$
Proof:

 Suppose $x \in \mathcal{N}$ is a critical point of $F\mid_{\mathcal{N}}$, i.e. $\mathrm{d}(F\mid_{\mathcal{N}})_x (\delta x) = 0$ for all $\delta x \in T_x\mathcal{N}$, or $\mathrm{d}F_x \in {T_x\mathcal{N}^0}^2.$ Since $\Phi \colon \mathcal{D} \to \mathcal{N}$ is bijective, there exists $e\in \mathcal{D}$ such that $\Phi(e)=x$ and for given $\delta e \in T_e\mathcal{D}$ there exists $\delta x \in T_{\Phi(e)}\mathcal{N}$ such that $\Phi_*(e)(\delta e)= \delta x$, where $\Phi_*(e) \colon T_e\mathcal{D}\to T_{\Phi(e)}\mathcal{N}$ denotes the tangent map to $\Phi$ at $e\in \mathcal{D}$. It therefore follows that since $\mathrm{d}(F \mid_{\mathcal{N}})_x(\delta x)=0$ for all $\delta x \in T_x\mathcal{N}$, $(\mathrm{d}F)( \Phi_*(e)(\delta e))=0$ for any $\delta e \in T_e\mathcal{D}$.
  
  Computing the differential of $\mathbb{F}$, we obtain,
\begin{equation}\label{dFbold}
\mathrm{d}\mathbb{F}_{(x,\lambda,e)} (\delta x, \delta \lambda, \delta e) = dF_x(\delta x) + \langle \delta\lambda, x - \Phi (e)\rangle + \langle \lambda, \delta x - \Phi_*(e) (\delta e) \rangle \, ,
\end{equation}
where  $\delta e \in T_e\mathcal{D}$, $\delta x \in T_x\mathcal{M} \cong \mathcal{M}$, $\delta \lambda \in T_\lambda \mathcal{M}^*\cong T_\lambda^*\mathcal{M} \cong \mathcal{M}^*$.  The notation $\langle \lambda , x\rangle$ denotes the natural pairing between $\mathcal{M}$ and its dual space $\mathcal{M}^*$.

 For $(x,\lambda,e)$ such that $x\in\mathcal{N}$ is a critical point of $\mathcal{F}\mid_{\mathcal{N}}$, $\Phi(e)=x$ and $\lambda = -\mathrm{d}F_x\in T_x\mathcal{N}^0 \subset T_x^*\mathcal{M}\cong\mathcal{M^*}$, it follows from $(3.1)$ and from the prior statements that $\mathrm{d}\mathbb{F}_{(x,\lambda,e)}(\delta x,\delta \lambda, \lambda e)= 0$ for all $\delta x,\delta\lambda$ and $\delta e$. Thus $(x,\lambda,e)$ is a critical point of $\mathbb{F}$.

 Now we prove the other direction of the theorem.
Let $(x,\lambda, e)$ be a critical point of $\mathbb{F}$, i.e. $\mathrm{d}\mathbb{F}_{(x,\lambda,e)}(\delta x,\delta \lambda, \delta e)=0$ for all $\delta x, \delta \lambda,\delta e.$ In particular, fixing $\delta x = \delta e =0$, for any $\delta\lambda$, since $(x,\lambda,e)$ is a critical point of $\mathbb{F}$, $\mathrm{d}\mathbb{F}_{(x,\lambda,e)}(0,\delta\lambda,0)= 0$. This implies by $(3.1)$ that $x=\Phi(e)$, thus $x\in \mathcal{N}.$ For any $\delta x\in T_x\mathcal{N}$, since $x=\Phi(e)$ and since $\Phi\colon \mathcal{D}\to\mathcal{N}$ is bijective, there exists $\delta e\in\mathcal{D}$ such that $\Phi_*(e)(\delta e)= \delta x$. So for our critical point $(x,\lambda,e)$ of $\mathbb{F}$ and for any $\delta x\in T_x\mathcal{N}$, applying $(3.1)$, we obtain $ 0 = \mathrm{d}\mathbb{F}_{(x,\lambda,e)}(\delta x,\delta \lambda,\delta e)= \mathrm{d}F_x(\delta x)$. Thus for any $\delta x \in T_x\mathcal{N},  \mathrm{d}F_x(\delta x)=0$, i.e. $x$ is a critical point of $F$.

%%%%%%%%%%%%%%%%%%%%%%%%%%%

\subsubsection{Critical points}

We apply Lagrange's multipliers theorem discussed in the previous section to the following setting.   The affine manifold $\mathcal{M}$ is the space of fields $J^1 \mathcal{F}_M^*$ in the covariant phase space .   The manifold $\mathcal{D}$ is the manifold of vierbein fields, i.e, sections $e$ of the bundle $F$ discussed before.
The map submanifold $\mathcal{N}$ is the submanifold $\mathcal{P}$ defined by Palatini's constraints, i.e., we have the map $P \colon \mathcal{D} \to J^1\mathcal{F}^*$ given by $P(e) = \epsilon \, e\wedge e$.   Then, finally, the function $F\colon \mathcal{M} \to \mathbb{R}$ is the topological Yang-Mills action functional $S_{YM,0} \colon J^1\mathcal{F}_M^*\to \mathbb{R}$.

Then we conclude that critical points of Palatini's action $S_P$ are in correspondence with families of critical points of the extended action:
$$
\mathbb{S} (A,P,\Lambda, e) = S_{YM,0}(A,P) + \langle \Lambda, P - \epsilon \, e\wedge e \rangle \, ,
$$
or, more explicitly:
\begin{equation}\label{paction}
\mathbb{S} (A,P,\Lambda, e) = \int_M P^{\mu\nu}_{IJ} F_{\mu\nu}^{IJ} + \Lambda_{\mu\nu}^{IJ}\left( P_{\mu\nu}^{IJ} - \epsilon\, e_\mu^I e_\nu^J \right) \mathrm{vol}_M \, .
\end{equation}

According to Lagrange's multipliers theorem, the critical points of $\mathbb{S}$ have the form $(A,P, \Lambda, e)$ where $(A,P)$ is a critical point of $S_{YM,0}\mid_{\mathcal{P}} = S_P$, for all $\Lambda$, i.e., $P = \epsilon \, e\wedge e$ for some vierbein field $e$ and $(A,e)$ is a critical point of:
$$
S_P = \int_M e_I^\mu e_\nu^J F_{\mu\nu}^{IJ} \epsilon \,  \mathrm{vol}_M \, .
$$
Then standard arguments (se for instance \cite{Ba**}, \cite{Gi**}) show that the Palatini connection is torsionless
and metric with respect to the metric $g_e$ defined by the vierbein field, that is $A$ is the Levi-Civita connection of the metric $g_e$.  Moreover, it satisfies Ricci's equation:
$$
\mathrm{Ric}(A) = 0 \, .
$$

From Eq. (\ref{dFbold}) we also get that if $(x,\lambda, e)$ is a critical point of $\mathbb{F}$, then at $x\in \mathcal{N}$ we get:
$$
\mathrm{d}F_x (\delta x) = - \langle \lambda, \delta x - \Phi_*(e) \delta e \rangle \, ,
$$
and $\delta x$ an arbitrary vector in $T_x\mathcal{M}$, that is not necessarily in $T_x\mathcal{N}$.  
This shows that if $(A,P = P(e), \Lambda, e)$ is a critical point of $\mathbb{S}$, then
$$
\mathrm{d}S_P (A,P = P(e))  (\delta A, \delta P) = - \langle \Lambda, \delta P - P_*(e) \delta e \rangle.
$$

%%%%%%%%%%%%%%%%%%%%%%%%%%%%%%%%%%%%%%%%%%%%%%%%%%%
%%%%%%%%%%%%%%%%%%%%%%%%%%%%%%%%%%%%%%%%%%%%%%%%%%%

\subsection{The canonical formalism near the boundary}\label{sec:canonical}
In order to obtain an evolution description for Palatini Gravity and to prepare the ground for canonical quantization, we need to introduce a local time parameter. We will only assume that a collar $U_{\epsilon} = (-\epsilon, 0]\times \partial M$ around the boundary can be chosen and so that a choice of a time parameter $t = x^0$ can be made near the boundary that would be used to describe the evolution of the system.  The fields of the theory would then be considered as fields defined on a given spatial frame that evolve in time for $t \in (-\epsilon, 0]$.  

The dynamics of such fields would be determined by the restriction of the Palatini action \eqref{paction} to the space of fields 
on $U_{\epsilon}$. Expanding we obtain,
\begin{eqnarray*}
S(A,P,\Lambda,e)
 & = & \int_{U_\epsilon}[P_{IJ}^{\mu\nu}F_{\mu\nu}^{IJ} + \Lambda_{\mu\nu}^{IJ}(P_{\mu\nu}^{IJ} -{\epsilon}e_{\mu}^Ie_{\nu}^J)]vol_M \\
& = & \int_{U_{\epsilon}}[ P_{IJ}^{\mu\nu}(-\frac{1}{2})(\partial_{\mu}A_{\nu}^{IJ}-\partial_{\nu}A_{\mu}^{IJ} \\
& + & \epsilon_{KL,MN}^{IJ}A_{\mu}^{KL}A_{\nu}^{MN})+ \Lambda_{\mu\nu}^{IJ}(P_{\mu\nu}^{IJ}- \epsilon e_{\mu}^Ie_{\nu}^J)]vol_M \\
& = & \int_{-\epsilon}^0 dt \int_{\partial M}vol_{\partial M}[P_{IJ}^{k0}(\partial_0 A_k^{IJ}- \partial_kA_0^{IJ} + \epsilon_{KL,MN}^{IJ}A_0^{KL}A_k^{MN})\\
& &-\frac{1}{2} P_{IJ}^{kj}(\partial_kA_j^{IJ}-\partial_jA_k^{IJ} + \epsilon_{KL,MN}^{IJ}A_k^{KL}A_j^{MN})
+ 2{\Lambda_{k0}^{IJ}}(P_{k0}^{IJ}-\epsilon e_k^Ie_0^J)+ \Lambda_{kj}^{IJ}(P_{kj}^{IJ}-\epsilon e_k^Ie_j^J)].
\end{eqnarray*}

In the previous expressions $\epsilon_{bc}^a$ denote the structure constants of the Lie algebra $\mathfrak{g}$ with respect to the basis $\xi_a$, that is $[\xi_b, \xi_c] = \epsilon_{bc}^a \xi_a$.
Notice that  $\epsilon^a_{bc}A^b_0A^c_0=0$ because for fixed a, ${\epsilon}^a_{bc}$ is skew-symmetric.  Moreover the indexes $\mu$ and $a$ have been pushed down and up by using the metric $\eta$ and the Killing-Cartan form $\langle\cdot, \cdot \rangle$ respectively.

 In the last equation we used that $P$ is a bivector, i.e., $P_a^{\mu\nu}$ is skew symmetric in $\mu$ and $\nu$ and therefore $P_a^{00} = 0$, and also $P^{k0}_a P^a_{k0} = P^{0i}_a P^a_{0i}$, because $P^{k0} = - P^{0k}$, etc. The momenta fields are defined as sections of the bundle $P(E)$ and as such are unrestricted. However, because Yang-Mills theories are Lagrangian theories the Legendre transform selects a subspace of the space of momenta that corresponds to fields $P$, skew symmetric in the indices $\mu$, $\nu$.(For more details see [Ib15].)

The previous expression acquires a clearer structure by introducing the appropriate notations for the fields restricted at the boundary and assuming that they evolve in time $t$.
Thus the pull-backs of the components of the fields $A$ and $P$ to the boundary will be denoted respectively as,
\begin{eqnarray*}
a^a_k &:=& A^a_k\mid_{\partial M} ; \qquad  a = (a^a_k) \, , \qquad a^a_0 := A^a_0\mid_{\partial M} ; \qquad a_0 = (a^k_0) \, , \\
p^k_a &:=& P^{k0}_a \mid_{\partial M} ; \qquad p = (p^k_a)\, , \qquad p^0_a := P^{00}_a\mid_{\partial M}= 0 ; \qquad p_0 = (p^0_a)=0 \, , \\
\beta^{ki}_a &:=& P^{ki}_a\mid_{\partial M} ; \qquad \beta =( \beta^{ki}_a) \, .
\end{eqnarray*}
Given two fields at the boundary, for instance $p$ and $a$, we will denote as usual by $\langle p, a\rangle$ the expression, 
$$
\langle p, a\rangle = \int_{\partial M} p_a^\mu a_\mu^a \, \mathrm{vol}_{\partial M}\, ,
$$
and the contraction of the inner (Lie algebra) indices by using the Killing-Cartan form and the integration over the boundary is understood. 

Introducing the notations and observations above in the expression for $S_{ U_\epsilon}$ we obtain,
\begin{eqnarray*}
S_{U_{\epsilon}}(A,P,\Lambda,e)&=& \int_{-\epsilon}^0dt\int_{\partial M} vol_{\partial M} [p_{IJ}^k(\partial_0a_k^{IJ}-\partial_ka_0^{IJ} + \epsilon_{KL,MN}^{IJ}a_0^{KL}a_k^{MN})\\
& & -\frac{1}{2} \beta_{IJ}^{kj}(\partial_ka_j^{IJ}-\partial_ja_k^{IJ}+ \epsilon_{KL,MN}^{IJ}a_k^{KL}a_j^{MN})
+ 2 \Lambda_{k0}^{IJ}(p_k^{IJ}-\epsilon e_k^Ie_0^J) + \Lambda_{kj}^{IJ}(\beta_{kj}^{IJ}-\epsilon e_k^Ie_j^J)].\\
&=& \int^0_{- \epsilon} \d t \, \mathcal{L}(a,\dot{a},a_0,\dot{a}_0,p,\dot{p}, \beta,\dot{\beta},\Lambda,\dot{\Lambda},\Lambda_0,\dot{\Lambda_0},e,\dot{e}) \,
 \end{eqnarray*}

where

$$\mathcal{L}(a,\dot{a},a_0,\dot{a_0},p,\dot{p},\beta,\dot{\beta},\Lambda,\dot{\Lambda},\Lambda_0,\dot{\Lambda_0},e,\dot{e})= <p,\dot{a}-d_aa_0 + 2\Lambda_{0}>$$
$$ - <\beta,F_a-\Lambda>+
<\Lambda_0,-2\epsilon e\wedge e_0>
 + <\Lambda, -\epsilon e\wedge e>. $$

Euler-Lagrange equations will have the form:
$$
\frac{d}{\d t} \frac{\delta \mathcal{L}}{\delta \dot{\chi}} = \frac{\delta \mathcal{L}}{\delta \chi} \, ,
$$
where $\chi \in P(E)$ and $\delta /\delta \chi$ denotes the variational derivative of the functional $\mathcal{L}$. \\\\

Thus for $\chi = p$ we obtain,
$$
\frac{\delta \mathcal{L}}{\delta \dot{p}}=0, \quad \mathrm{hence} \quad 0 = \frac{\delta \mathcal{L}}{\delta p} = \dot{a} - \d_a a_0 + 2\Lambda_0 \, ,
$$
and thus,
 \begin{equation}\label{adot}
 \dot{a} =  \d_a a_0  -  2\Lambda_0 \, .
 \end{equation}  \\
 
For $\chi = a$ we obtain,
$$
\frac{\delta \mathcal{L}}{\delta \dot{a}}=p, \quad \mathrm{hence} \quad \dot{p} =\frac{\delta \mathcal{L}}{\delta a} = d^*\beta + [p,a_0]  \, ,
$$
that is,
 \begin{equation}\label{pdot}
 \dot{p} =  d^*\beta + [p,a_0]. 
 \end{equation}  \\  
 
For $\chi = a_0$ we obtain,
 $$
 \frac{\delta \mathcal{L}}{\delta \dot{a_0}}=0, \quad \mathrm{hence} \quad 0 = \frac{\delta \mathcal{L}}{\delta a_0}= \frac{\partial}{\partial a_0}<p,d_aa_0> = \frac{\partial}{\partial a_0}-<d^*_ap,a_0> = -d_a^*p \, ,
 $$
that is,
 \begin{equation}\label{constraint1}
 d_a^*p = 0 .
 \end{equation} \\
 
For $\chi = \beta $ we obtain,
 $$
 \frac{\delta \mathcal{L}}{\delta \dot{\beta}}=0, \quad \mathrm{hence} \quad 0= \frac{\delta \mathcal{L}}{\delta \beta}= -F_a + \Lambda, \,
 $$
that is,
\begin{equation}\label{constraint2}
F_a = \Lambda .
\end{equation} \\

For $\chi = \Lambda $ we obtain,
 $$
 \frac{\delta \mathcal{L}}{\delta \dot{\Lambda}}=0, \quad \mathrm{hence} \quad 0 = \frac{\delta \mathcal{L}}{\delta \Lambda}= \beta -\epsilon e\wedge e, \,
 $$
 that is,
 \begin{equation}\label{constraint3}
 \beta = \epsilon e \wedge e .
 \end{equation} \\
 
For $\chi = \Lambda_0 $ we obtain,
 $$
 \frac{\delta \mathcal{L}}{\delta \dot{\Lambda_0}}=0, \quad \mathrm{hence} \quad 0 = \frac{\delta \mathcal{L}}{\delta \Lambda_0}= 2p - 2\epsilon e \wedge e_0, \, 
 $$
 that is,
 \begin{equation}\label{constraint4}
 \mathrm{p} = \epsilon e \wedge e_0 .
 \end{equation}
For $\chi = e $ we obtain,
 $$
\frac{\delta \mathcal{L}}{\delta \dot{e}}=0, \quad \mathrm{hence} \quad 0 = \frac{\delta \mathcal{L}}{\delta e}= -2\epsilon e_0 \Lambda_0 - 2 \epsilon e \Lambda, \,
 $$
 that is,
 \begin{equation}\label{constraint5}
 - e_0 \Lambda_0 = e \Lambda .
 \end{equation}
For $\chi = e_0 $ we obtain,
 $$
 \frac{\delta \mathcal{L}}{\delta \dot{e_0}}=0, \quad \mathrm{hence} \quad 0 = \frac{\delta \mathcal{L}}{\delta e_0}= 2\epsilon e\Lambda_0
 $$
 i.e.,
 \begin{equation}\label{constraint6}
 e\Lambda_0 = 0.
 \end{equation}
Thus solving for the Euler-Lagrange equations, we have obtained two evolution equations, \eqref{adot} and \eqref{pdot} and six constraint equations \eqref{constraint1} - \eqref{constraint6}.

\subsection{The presymplectic formalism: Palatini at the boundary and reduction}
  As discussed in general in section $3.2$, we define the extended Hamiltonian, $\mathcal{H}$, so that $\mathcal{L} = \langle p, \dot{a}\rangle - \mathcal{H}$ :
  
\begin{equation}\label{PH}
\mathcal{H}(a,a_0,p,\beta,\Lambda,\Lambda_0,e)= <p,-d_aa_0 + 2\Lambda_{0}>
 - <\beta,F_a-\Lambda>+
<\Lambda_0,-2\epsilon e\wedge e_0> \quad\\\\
\quad \quad + \langle\Lambda,-\epsilon e\wedge e \rangle.
\end{equation}

Thus the Euler-Lagrange equations can be rewritten as 
\begin{equation}\label{PMP1}
\dot{a} = \frac{\delta \mathcal{H}}{\delta p}; \quad \dot{p} = - \frac{\delta \mathcal{H}}{\delta a} \, ,
\end{equation}
\begin{equation}\label{Constraints}
\frac{\delta \mathcal{H}}{\delta a_0} = 0 ; \quad \frac{\delta \mathcal{H}}{\delta \beta} = 0 ; \quad
\frac{\delta \mathcal{H}}{\delta \Lambda} = 0 ; \quad
\frac{\delta \mathcal{H}}{\delta \Lambda_0} = 0 ; \quad
\frac{\delta \mathcal{H}}{\delta e} = 0 ; \quad
\frac{\delta \mathcal{H}}{\delta e_0}= 0.
\end{equation}

We denote again by $\varrho \colon \mathcal{M} \to T^*\mathcal{F}_{\partial M}$ the canonical projection $\varrho(a,a_0,p,\beta)=(a,a_0,p)$.  Let $\omega_{\partial M}$ denote the form on the cotangent bundle $T^*\mathcal{F}_{\partial M}$, 
$$
\omega_{\partial M} =  \delta a \wedge \delta p . 
$$
We will denote again by $\Omega$ the pull-back of this form to $\mathcal{M}$ along $\varrho$, i.e., $\Omega = \varrho^*\omega_{\partial M}$.  Clearly, $\ker {\Omega} = \mathrm{span} \{ \delta /\delta \beta, \delta/\delta a_0 \}$, and we have the particular form that Thm.  \ref{presymplectic_equation} takes here.

\begin{theorem}
The solution to the equation of motion defined by  the Palatini Lagrangian \eqref{paction}, 
are in one-to-one correspondence with the integral 
curves of the presymplectic system $(\mathcal{M},\Omega,\mathcal{H})$, i.e. with the integral curves of the vector field $\Gamma$ on $\mathcal{M}$ such that $i_\Gamma\Omega= \mathrm{d} \mathcal{H}$.
\end{theorem}          

The primary constraint submanifold $\mathcal{M}_1$ is
 defined by the six constraint equations,
$$
\mathcal{M}_1= \{(a,a_0,p,\beta, \Lambda,e)|\mathcal{F}_a = \Lambda, d_a^*p = 0, \beta = \epsilon e\wedge e, p = \epsilon e \wedge e_0, e_0\Lambda_0 = e\Lambda, e \Lambda_0 = 0 \} \, .
$$
\\ 
Since $\Lambda = F_a$, and $\beta$ is a just a function of $e$,  we have that \\
 \\
 $ \mathcal {M}_1\cong \{(a,a_0,p,e)| d_a^*p = 0 , p= \epsilon e \wedge e_0, e_0 {F_a}_0 = e F_a , e {F_a}_0 = 0 \}$\\ 
 
 and 
\quad $\ker \Omega|_{\mathcal{M}_1} \supset \mathrm{span} \{\frac{\partial}{\partial a_0}\}$.

 Thus
$\mathcal{M'}_1 = \mathcal{M}_1/(\ker \Omega|\mathcal{M}_1) \cong \{(a,p,e)| \mathrm{d}_a^*p = 0, p=\epsilon e \wedge e_0, e_0{F_a}_0 = eF_a, e{F_a}_0 = 0 \}.$

%%%%%%%%%%%%%%%%%
%%%%%%%%%%%%%%%%%

\subsection{Gauge transformations: symmetry and reduction}

The group of gauge transformations $\mathcal{G}$, i.e, the group of automorphisms of the principal bundle $P$ over the identity, is a fundamental symmetry of the theory.      Notice that the Palatini action is invariant under the action of $\mathcal{G}.$

The quotient of the group of gauge transformations by the normal subgroup of identity gauge transformations at the boundary defines the group of gauge transformations at the boundary $\mathcal{G}_{\partial M}$, and it constitutes a symmetry group of the theory at the boundary, i.e. it is a symmetry group both of the boundary Lagrangian $\mathcal{L}$ and of the presymplectic system $(\mathcal{M}, \Omega, \mathcal{H})$.    We may take advantage of this symmetry to provide an alternative description of the constraints found in the previous section.    

\begin{proposition}
The map $\mathcal{J} \colon T^*\mathcal{F}_{\partial M} \to \mathfrak{g}_{\partial M}^*$ given by                                         $\mathcal{J}(a,p) = \-d_a^*p$ is the moment map of the action of the group $\mathcal{G}_{\partial M}$ on $ T^*\mathcal{F}_{\partial M}$ where the action of $\mathcal{G}_{\partial M}$ on $ T^*\mathcal{F}_{\partial M}$ is by cotangent liftings.
\end{proposition} 

\begin{proof}  The moment map $\mathcal{J} \colon T^*\mathcal{F}_{\partial M} \to \mathfrak{g}_{\partial M}^*$ is given by,
$$
\langle \mathcal{J} (a,p), \xi \rangle = \langle p, \xi_{\mathcal{F}_{\partial M}}\rangle = \langle
 p, \d_a\xi \rangle = \langle -d_a^*p, \xi \rangle \, ,
$$
because the gauge transformation $g_s = \exp s \xi$ acts in $a$ as $a \mapsto g_s\cdot a =   g_s^{-1} a g_s + g_s^{-1}\d g_s$
  and
the induced tangent vector is given by,
$$
\xi_{\mathcal{A}_{\partial M}} (a)= \frac{\d}{\d s}  g_s \cdot a \mid_{s = 0} = \d_a\xi \, .
$$

\end{proof}

By the standard Marsden-Weinstein reduction, $\mathcal{J}^{-1}(\mathbf{0})=\{(a,p)\in \mathcal{T^*F}_{\partial M}|  d_a^*p = 0\}$ is a coisotropic submanifold of the symplectic manifold $\mathcal{T^*F}_{\partial M}$ and $\mathcal{J}^{-1}(\mathbf{0})/\mathcal{G}_{\partial M}$ is symplectic.
$\{(a,p)\in \mathcal{T^*F}_{\partial M}| p = e\wedge e_0\}$ is easily seen to be a symplectic submanifold of $(\mathcal{T^*F}_{\partial M}), \Omega), \Omega = \delta a \wedge \delta p$. $ e_0{F_a}_0 = eF_a$ and $e{F_a}_0 = 0$ are coisotropic submanifolds of $\mathcal{T^*F}_{\partial M}$. This follows from the elementary observation that in a symplectic manifold a subspace defined by a function, $\phi = 0$ is a coisoptropic submanifold of the symplectic manifold.
 Now we need to check that the intersection of the coisotropic submanifolds comprising $\mathcal{M}_1'$ is a coisotropic submanifold. But this follows easily from the fact that the kernel  $\mathcal{J}^{-1}(\mathbf{0})=\{(a,p)\in \mathcal{T^*F}_{\partial M}|d_a^*p = 0\}$ is spanned by the action of the gauge group $\mathcal{G}_{\partial M}$ and from the observation that the action of $\mathcal{G}_{\partial M}$ leaves invariant the submanifolds $e_0{F_a}_0 = eF_a$ and $e{F_a}_0=0$. ker $\mathcal{J}^{-1}(\mathbf{0})$ is tangent to $\{(a,e)| e{F_a}_0=0\}$ and to $\{(a,e)| e_0{F_a}_0 = eF_a\}$ and is therefore contained in the tangent spaces of the two surfaces, and vice versa. Thus $\mathcal{M}_1'$ is a coisotropic and as described in section 3.2, the reduced space $\mathcal{R} = \mathcal{M}_1'/\mathcal{G}_{\partial M}$ is symplectic and $\widetilde{\Pi}(\mathcal{EL})$ is an isotropic submanifold of $\mathcal{R}$.

 \section{Conclusions and discussion}
 
  Using multisymplectic geometry we have described a Hamiltonian formulation of Palatini's General Relativity that is simple. Unlike ADM it does not involve lapse and shift operators and it does not require for it's application the assumption that spacetime is topologically ${R} \times S$ where $S$ is space. All we need to assume is that our spacetime manifold has a boundary and that the boundary has a collar.\\ 
  After the presymplectic constraint analysis, the analysis in the collar provides consistent solutions of the initial value problem for General Relativity. Unlike ADM, we use a formalism that is canonical, i.e. at every step the fundamental structures are preserved, both when discussing the constraints introduced from the bulk Palatini constraint $P= e \wedge e$ and when reducing the system by using gauge invariance.\\
   In a following work we will apply our techniques to study Ashtekar gravity and to the corresponding quantum aspects.

\section*{Acknowledgements}\label{sec:acknowledgements}

A.S. thanks Nicolai Reshetikhin for suggesting to her a problem that motivated this work.
A.I. was partially supported by the Community of Madrid project QUITEMAD+, S2013/ICE-2801, and by MINECO grant MTM2014-54692-P.

\end{document}